\documentclass[journal,12pt,onecolumn]{IEEEtran}
\usepackage{mathrsfs}
\usepackage{pdfpages}
\usepackage{graphicx,color}
\usepackage{amsmath}
\usepackage{epstopdf}
\usepackage{float}
\usepackage{amssymb} 
\usepackage[noadjust]{cite}
\usepackage{caption}
\usepackage{subcaption}
\usepackage{bbm}
\usepackage{amsthm,cite}

\newtheorem{thm}{Theorem}[section]

\newtheorem{lem}[thm]{Lemma}

\newtheorem{Def}{Definition}

\newcommand\wto{\xrightarrow[]{w}}

\newtheorem{prob-statement}{Problem}

\DeclareMathOperator*{\arginf}{arg\,inf}

\begin{document}
	\title{Robust Kullback-Leibler Divergence and Universal Hypothesis Testing for Continuous Distributions}
	\author{\IEEEauthorblockN{Pengfei Yang and Biao Chen\\}
\thanks{This work was supported in part by the Air Force Office of
Scientific Research under Grant FA9550-16-1-0077. 

P. Yang was with the Department of Electrical Engineering and Computer Science, Syracuse University, Syracuse,
NY, 13244 USA. Email: ypengf@gmail.com. He is now with Point72 Asset Management, L.P., New York, NY. B. Chen is with the Department
of Electrical Engineering and Computer Science, Syracuse University, Syracuse,
NY, 13244 USA. Email: bichen@syr.edu.}
	}
	\maketitle
	
	\begin{abstract}
Universal hypothesis testing refers to the problem of deciding whether samples come from a nominal distribution or an unknown distribution that is different from the nominal distribution. Hoeffding's test, whose test statistic is equivalent to the empirical Kullback-Leibler divergence (KLD), is known to be asymptotically optimal for distributions defined on finite alphabets. With continuous observations, however, the discontinuity of the KLD in the distribution functions results in significant complications for universal hypothesis testing. This paper introduces a robust version of the classical KLD, defined as the KLD from a distribution to the L$\acute{e}$vy ball of a known distribution. This robust KLD is shown to be continuous in the underlying distribution function with respect to the weak convergence. The continuity property enables the development of a universal hypothesis test for continuous observations that is shown to be asymptotically optimal for continuous distributions in the same sense as that of the Hoeffding's test for discrete distributions.
	\end{abstract}

	\section{Introduction}

The Kullback-Leibler divergence (KLD), also known as the relative entropy, is one of the most fundamental metrics in information theory and statistics \cite{Cover:1991, ICPCS_2004}. 
The KLD has a number of operational meanings and finds applications in a diverse range of disciplines. For example, the mutual information, which is a special case of the KLD, is a fundamental quantity in both channel coding and data compression \cite{Cover:1991}. In hypothesis testing, the KLD is known to be the decay rates of error probabilities (e.g., see Stein's lemma \cite{Cover:1991} and Sanov's theorem \cite{largebook}). 
	
An important application of the KLD is in the so-called universal hypothesis testing: given a nominal distribution $P_0$, the objective is to decide, upon observing a sample sequence, whether the underlying distribution that generates the sequence is $P_0$ or a distribution {\em different} from $P_0$. This problem was first formulated by Hoeffding \cite{Asymptotically}; with finite alphabet, Hoeffding developed a detector that is shown to be optimal according to the generalized Neyman-Pearson (NP) criterion, i.e., it achieves optimal type II error exponent subject to a constraint on the type-I error exponent \cite{Asymptotically}. The test statistic of Hoeffding's detector is equivalent to the KLD between the empirical distribution and $P_0$. 

Hoeffding's result, however, does not generalize to the universal hypothesis testing with continuous alphabet. Clearly, computing empirical KLD 
for continuous distributions is meaningless as the empirical distribution, 
which is discrete, and the nominal distribution $P_0$, which is continuous, have different support sets. Additionally, the asymptotic optimality of Hoeffding's test 
was established using a combinatorial argument \cite{Asymptotically} and thus is inapplicable to the continuous case. Attempts to reconstruct a similar decision rule for continuous observations have been largely fruitless with the only exception of the work by Zeitouni and Gutman \cite{largedevi} where large deviation bounds were used in lieu of combinatorial bounds. The results in \cite{largedevi}, however, are obtained at the cost of a weaker optimality with a rather complicated detector.

The difficulty in dealing with continuous observations for universal hypothesis testing stems from the subtle but important distinction on the continuity property of the KLD with respect to the underlying distributions. With finite alphabet distributions, the KLD defined between two distributions is known to be continuous in the distribution functions. This is not the case for the KLD defined between two distributions on the real line, i.e., those with continuous observations \cite{lsc_}. Specifically, weak convergence (i.e., convergence of distribution functions) does not imply convergence of the KLD. As such, even when two distributions who are arbitrarily close in terms of distribution functions, the KLD between them can be arbitrarily large.

This paper defines a robust version of the classical KLD that utilizes the L$\acute{e}$vy metric which, unlike the KLD, is a true distance metric for distributions. 
The robust KLD, defined as the KLD from a distribution to a L$\acute{e}$vy ball of another distribution, is shown to be continuous in the first distribution. This continuity property enables the development of a test for the universal hypothesis testing that is similar in its form to Hoeffding's detector and attains the desired asymptotic NP optimality. Not only is the optimality in a stronger sense than that of \cite{largedevi}, but the test statistic is also much more intuitive and easier to compute.

	The rest of the paper is organized as follows. Section II defines the KLD between two (sets of) distributions; introduces the concepts of weak convergence and the L$\acute{e}$vy metric along with their connections; and review the universal hypothesis testing for the finite alphabet case. 
Section III defines the robust KLD and establishes the continuity property with respect to weak convergence.  In Section IV, the large deviation approach by Zeitouni and Gutman \cite{largedevi} to the universal hypothesis testing for continuous distributions is first reviewed; a robust version of the universal hypothesis testing problem is then introduced and the asymptotically NP optimal test using the robust KLD is derived. Section V concludes this paper.

\section{Preliminaries}

\subsection{The KL divergence}

The KLD was first introduced in \cite{SKRAL:1951}. It is often used as a metric to quantify the distance between two probability distributions. For finite alphabets, the KLD between a probability distribution $\mu = (\mu_1, \mu_2, \cdots, \mu_n)$ and another distribution $P = (p_1, p_2, \cdots, p_n)$ is 
	\begin{align}
	D(\mu||P) = \sum_{i=1}^n \mu_i\log\frac{\mu_i}{p_i}. \label{eq:kld1}
	\end{align}
	For distributions defined on the real line $\mathcal{R}$, the KL divergence between $\mu$ and $P$ is defined as
	\begin{eqnarray}
	D(\mu||P) = \int_{\mathcal{R}}d\mu\log\frac{d\mu}{dP}. \label{def:kl2}
	\end{eqnarray}
	
The KLD $D(\mu||P)$ is jointly convex for both discrete and continuous distributions. For two sets of probability distributions, say $\Gamma_1$ and $\Gamma_2$, defined on the same probability space,  the KLD between the two sets is defined to be the infimum of the KLD of all possible pairs of distributions, i.e.,
\begin{equation}
D(\Gamma_1||\Gamma_2):=\inf_{\gamma_1\in\Gamma_1, \gamma_2\in\Gamma_2}D(\gamma_1||\gamma_2). \label{eq:KLDsets}
\end{equation}
	
\subsection{Weak convergence and the L\'evy metric}

Denote the space of probability distributions on $(\mathcal{R},\mathcal{F})$ as $\mathcal{P}$, where $\mathcal{R}$ is the real line and $\mathcal{F}$ is the sigma-algebra that contains all the Borel sets of $\mathcal{R}$. For $P\in\mathcal{P}$, $P(S)$ is defined for the set $S\in\mathcal{F}$. A clear and simple notation commonly used is $P(t):=P((-\infty,t])$, since $P$ and its corresponding cumulative distribution function (CDF) are equivalent, i.e., one is uniquely determined by the other \cite{spaceproperty}.

Weak convergence is defined to be the convergence of the distribution functions as given below.
	\begin{Def}
		(Weak convergence \cite{pthm, spaceproperty})
		For $P_n$, $P\in\mathcal{P}$, we say $P_n$ weakly converges to $P$ and write $P_n\wto P$, if $P_n(x)\to P(x)$ for all $x$ such that $P$ is continuous at $x$. 
	\end{Def}

	The L$\acute{e}$vy metric $d_L$ between distributions $F\in\mathcal{P}$ and $G\in\mathcal{P}$ is defined as
\[
d_L(F, G):= \inf\{\epsilon: F(x-\epsilon)-\epsilon\le G(x)\le F(x+\epsilon)+\epsilon, \forall x\in\mathcal{R}\}.
\]
	The L$\acute{e}$vy metric makes $(\mathcal{P},d_L)$ a metric space \cite{largebook}, i.e., we have, for $\mu,P,Q\in \mathcal{P}$,
	\begin{eqnarray}
	d_L(\mu,P)=0 &\Leftrightarrow& \mu=P, \nonumber \\
	d_L(\mu,P)&=&d_L(P,\mu), \nonumber \\
	d_L(\mu,P)&\le&d_L(\mu,Q) + d_L(Q, P). \nonumber
	\end{eqnarray}
	The L$\acute{e}$vy ball centered at $P_0\in\mathcal{P}$ with radius $\delta$ is denoted as 
	\begin{eqnarray}
	B_L(P_0,\delta)=\{P\in\mathcal{P}: d_L(P,P_0)\le\delta\}. \label{eq:the levy ball definition}
	\end{eqnarray}
	\begin{figure}[thb!]
		\begin{center}
			\includegraphics[width=3.5in]{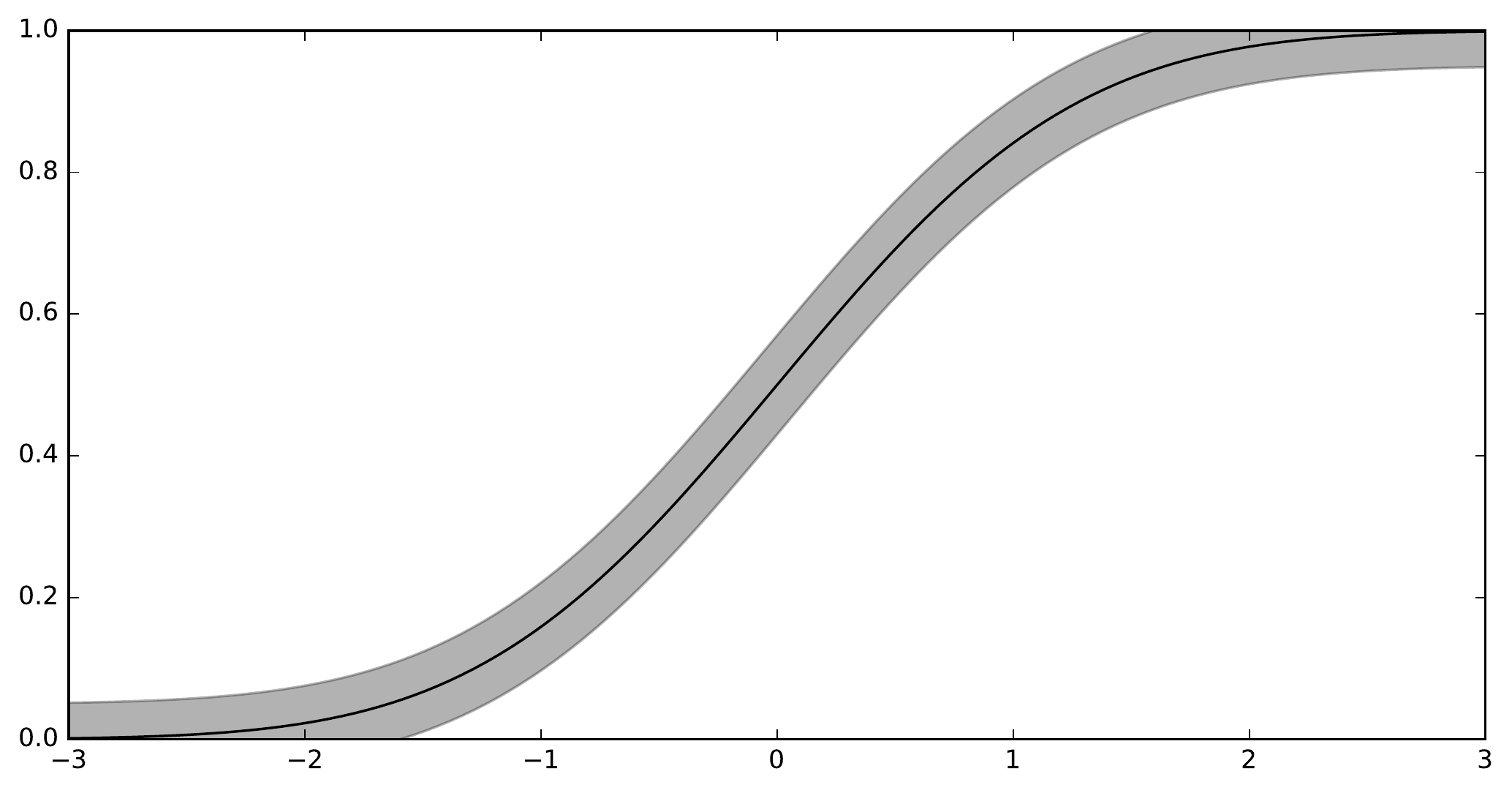}
			\caption[The L$\acute{e}$vy ball centered at standard normal distribution with radius $0.045$.]{The L$\acute{e}$vy ball centered at standard normal distribution with radius $0.045$.}{\label{fig:normalandlevy}}\label{fig:levy}
		\end{center}
	\end{figure}
	Fig.~\ref{fig:levy} plots the CDF of the standard normal distribution and its L$\acute{e}$vy ball with radius $0.045$. A distribution falls inside the shaded area if and only if its distance to the standard normal distribution, as measured by the Levy metric $d_L$, is less than or equal to $0.045$. 
	
		The L$\acute{e}$vy metric is strongly related to the concept of the weak convergence of probability measures. 
	\begin{lem}
		\cite{pthm, spaceproperty} For sequences in $\mathcal{P}$ whose limit is also in $\mathcal{P}$, the weak convergence and convergence in the $d_L$ are equivalent, i.e., if $(P_n\in\mathcal{P})$ is a sequence in $\mathcal{P}$ and $P\in\mathcal{P}$, then $P_n\wto P$ iff $d_L(P_n,P)\to0$. 
	\end{lem}

\subsection{Universal hypothesis testing}
Let a sequence of independent and identically distributed (i.i.d.) observations $(x_0, \cdots, x_{n-1})=x^n$ be the output of a source $P$.
Consider the following hypothesis test
	\begin{eqnarray}
	\mathcal{H}_0: P=P_0, \ \ \ \ \mathcal{H}_1: P=P',  \label{eq:singleUHT}
	\end{eqnarray}
	where $P_0$ is a known distribution while $P'\neq P_0$ is defined on the same probability space as $P_0$ but is otherwise unknown. The fact that $P'$ can be an arbitrary distribution gives rise to the name universal hypothesis testing. Clearly, while any decision rule will be independent of $P'$, the performance of the decision rule depends on $P'$.

	The universal hypothesis testing was first studied by Hoeffding \cite{Asymptotically} who considered distributions with finite alphabet. 
Hoeffding's detector is equivalent to the following threshold test of the empirical KLD:
	\begin{eqnarray}
	D(\hat{\mu}_n||P_0)\underset{{H_0}}{\overset{H_{1}}{\gtrless}}\eta. \label{detector:Hoeffding}
	\end{eqnarray}
Resorting to combinatorial bounds, Hoeffding successfully established the asymptotical Neyman-Pearson (NP) optimality of the above test. Specifically, let $\phi$ be the sequence of detectors $\{\phi^n(x_0,\cdots,x_{n-1}), n\ge1\}$. Define the error exponents for the two types of error probabilities respectively as follows,
	\begin{eqnarray}
	I^{P'}(\phi):=\liminf_{n\rightarrow\infty}-\frac{1}{n}\log P'^n(\phi^n(x^n)=0),  \nonumber\\
	J^{P_0}(\phi):=\liminf_{n\rightarrow\infty}-\frac{1}{n}\log P_{0}^n(\phi^n(x^n)=1). \nonumber
	\end{eqnarray}
	
Zeitouni and Gutman \cite{largedevi} have shown that to achieve the best trade-off between $I^{P'}$ and $J^{P_0}$, the test can depend on $x^n$ only through the empirical measure $\hat{\mu}_n$, defined to be 
\begin{eqnarray}
\label{eq:empirical}
\hat{\mu}_n(t)=\frac{\sum_iI_{\{x_i\leq t\}}}{n}. \label{eq:empirical}
\end{eqnarray}

The sequence of detectors can be equivalently expressed using $\Omega$ which is a  sequence of partitions $(\Omega_0(n), \Omega_1(n))\ (n=1,2, \cdots)$ of which $\Omega_0(n)\cap\Omega_1(n)=\emptyset$ and $\mathcal{P}=\Omega_0(n)\cup\Omega_1(n)$. The decision rule is made in favor of $H_i$ if $\hat{\mu}_n\in\Omega_i(n)$, $i=0, 1$.   Therefore $I^{P'}(\phi)$ and $J^{P_0}(\phi)$ can be written as
 	\begin{eqnarray}
 	I^{P'}(\Omega)=\liminf_{n\rightarrow\infty}-\frac{1}{n}\log P'^n(\hat{\mu}_n\in\Omega_0(n)),  \nonumber \\
 	J^{P_0}(\Omega)=\liminf_{n\rightarrow\infty}-\frac{1}{n}\log P_{0}^n(\hat{\mu}_n\in\Omega_1(n)).\nonumber
 	\end{eqnarray}

The generalized NP criterion maximizes the exponent of type II probability of error under a constraint on the minimal rate of decrease in type I probability of error:
	\begin{eqnarray}
	\label{eq:asmset}
	\max_{\Omega}I^{P'}(\Omega) \ \ \text{\ s.t.\ } \ 
	J^{P_0}(\Omega)\ge\eta. \label{eq:generalizedNP}
	\end{eqnarray}
Thus, among all sequences of detectors that satisfy the constraint on the type I error exponent in (\ref{eq:generalizedNP}), Hoeffding's test in (\ref{detector:Hoeffding}) maximizes the type II error exponent. 

%
	
\section{Robust KLD}
\subsection{Continuity property of the KLD}\label{subsec:continuityA}
Let the nominal distribution be $P_0$. Let $\mu$ and $\mu_k$, $k=1,2,\cdots$, be distributions with the same sample space as $P_0$. Suppose the sequence of distributions $\mu_k$ converge weakly to $\mu$. It is of interest to study whether the corresponding KLD between $\mu_k$ and $P_0$ also converge to the KLD between $\mu$ and $P_0$. That is,
does $\mu_k \wto \mu$ imply $D(\mu_k||P_0)\rightarrow D(\mu||P_0)$? 

The statement is true if the distributions involved are defined on a finite alphabet. With finite elements in the sample space of $P_0$, $D(\mu||P_0)$ is continuous in $\mu$ that has the same sample support as $P_0$. Convergence in distribution implies convergence in the corresponding KLD. 

This, however, is not the case for $P_0\in\mathcal{P}$. Indeed, it was established in  \cite{lsc_} that the KLD is only lower semicontinuous with respect to the weak convergence for the continuous case, i.e., 
\[
D(\mu || P_0) \leq \liminf_{k\rightarrow \infty} D(\mu_k || P_0).
\]
However, the KLD is not upper semicontinuous, i.e., the following result is not necessarily true:
\[
D(\mu || P_0) \geq \limsup_{k\rightarrow \infty} D(\mu_k || P_0),
\]
thus the KLD is not continuous in $\mu$ for continuous observations. To see this, let $\mu=P_0$ thus $D(\mu||P_0)=0$. Choose a distribution, say $P$, that is not absolutely continuous with respect to $P_0$.  Let $$P_{\epsilon}=(1-\epsilon)P_0+\epsilon P,$$ then $P_{\epsilon}$ is also not absolutely continuous with respect to $P_0$, thus $D(P_{\epsilon}||P_0)$ is unbounded for any given $\epsilon>0$. However, $P_{\epsilon}$ weakly converges to $P_0$ as $\epsilon\rightarrow 0$. While this construction takes advantage of the distributions that are not absolutely continuous with respect to $P_0$, the same is true even if one is constrained to a sequence of distributions that is absolutely continuous with respect to $P_0$.

This lack of continuity for the KLD for the continuous case is the primary reason for the difficulty in generalizing Hoeffding's result to the universal hypothesis testing with continuous observations. A direct consequence of the lack of continuity is that the superlevel set defined by the KLD is not closed \cite{largedevi}. The superlevel set is given by
	\begin{eqnarray}
	\{\mu\in\mathcal{P}: D(\mu||P_0)\ge\eta_1\}. \label{eq:disjoint2}
	\end{eqnarray}

The fact that the KLD is not continuous in $\mu$ leads to the unexpected property that the closure of the above set encompasses the entire probability space, i.e., any $P$ that does not belong to the above superlevel set has a sequence of distributions in the set that weakly converge to $P$. As such, a test in a similar form as (\ref{detector:Hoeffding}) can not be used for continuous distributions.

\subsection{Robust KLD and its continuity property}

Given a pair of distributions $(\mu,P_0)$, the robust KLD is defined to be the KLD between $\mu$ and the L$\acute{e}$vy ball centered at $P_0$. Using the definition (\ref{eq:KLDsets}), the robust KLD is the KLD defined between the two sets $\{\mu\}$ and $B_L(P_0,\delta_0)$ where $\delta_0>0$ is the radius of the L$\acute{e}$vy ball. We denote it simply as $D(\mu||B_L(P_0,\delta_0))$. 

The following theorem establishes its continuity property in $\mu$ under some mild assumptions. 
	
	\begin{thm}
		\label{thm:continuous}
		For a distribution $P_0\in\mathcal{P}$, if $P_0(t)$ is continuous in $t$, then for any $\delta_0>0$, $D(\mu||B_L(P_0,\delta_0))$ is continuous in $\mu$ with respect to the weak convergence. 
	\end{thm}
	The non-trivial part of the proof is to show that $D(\mu||B_L(P_0,\delta_0))$ is upper semicontinuous in $\mu$ (Lemma \ref{lem:uppersemi}). Lemma \ref{lem:lowersemi} proves $D(\mu||B_L(P_0,\delta_0))$ is lower semicontinuous in $\mu$. Therefore, $D(\mu||B_L(P_0,\delta_0))$ is continuous in $\mu$. The complete proof is lengthy and is included in Appendix \ref{app:continuous}. Important intermediate steps are summarized below.
	\begin{enumerate}
		\item
		We first partition (quantize) the real line into a set of finite intervals. The robust KLD corresponding to the quantized distributions converge to the true robust KLD as the quantization becomes finer. The proof is in essence proving that a max-min inequality is in fact an equality (Lemma A.1).
		\item
		The robust KLD is defined as the infimum over a L\'evy ball and it is established that there exists a distribution inside or on the surface of the L\'evy ball that achieves the infimum (see proof of Lemma A.1).
		\item
		The robust KLD is continuous in the radius of the L\'evy ball (Lemma A.2).
		\item
		The robust KLD and the quantized robust KLD are convex functions of the respective distributions (Lemma A.3). 
		\item
		The supremum of the robust KLD over a L\'evy ball centered at the first distribution is achieved by a distribution whose distribution function consists of two parts with a single transition point: the first part (i.e., prior to the transition point) corresponds to the lower bound of the L\'evy ball and the second part (i.e., after the transition point) corresponds to the upper bound of the L\'evy ball. Thus the class of distributions so defined is determined by the transition point given the  L\'evy ball. As such, the problem of finding an optimal distribution is reduced to finding an optimal transition point (Lemma A.4). 
		\item
		The robust KLD is bounded with  (Lemma A.5)
$$\sup_{\mu, P_0\in\mathcal{P}}D(\mu||B_L(P_0,\delta_0))=\log\frac{1}{\delta_0}.$$
		\item
		The supremum of the robust KLD over a L\'evy ball converges to the robust KLD as the L\'evy ball diminishes, i.e., as its radius goes to $0$. Therefore, the robust KLD is upper semicontinuous (Lemma A.6).
		\item
		The robust KLD is lower semicontinuous (Lemma A.7).
	\end{enumerate}
	The intuition of the continuity property of the robust KLD is the following. The classical KLD is a function of two distributions, and its value may vary arbitrarily large with small perturbation in one of the distributions with respect to the L\'evy metric. The reason is because the L\'evy metric is strictly weaker than the KLD, i.e., convergence in the KLD necessarily implies convergence in the L\'evy metric but {\em not} the other way around. 
For the robust KLD where the KLD is defined between the first distribution and a L\'evy ball centered around the second distribution, small perturbations in the first distribution can now be tolerated by the L\'evy ball around the second distribution, thanks again to the fact that the L\'evy metric is strictly weaker than the KLD.

\subsection{Discussions}

The continuity property in Theorem \ref{thm:continuous} does not hold if the distribution ball is constructed using some other measures, including the total variation and the KLD. 
	
	\begin{Def}
		\cite{TV} The total variation between $P\in\mathcal{P}$ and $Q\in\mathcal{P}$ is $d_{TV}(P,Q):=\sup_{S\in\mathcal{F}}|P(S)-Q(S)|$.
	\end{Def}
	
	Let $P_0(t)$ be continuous in $t$ and denote by $B_{TV}(P_0,\delta_0)$ and $B_{KL}(P_0,\delta_0)$ distribution balls defined using the total variation and the KLD, respectively, in a manner similar to that of the L\'evy ball in (\ref{eq:the levy ball definition}). We show that there exist a sequence $P_{n}$ that weakly converge to $P_0$, yet neither $D(P_{n}||B_{TV}(P_0,\delta_0))$ nor $D(P_{n}||B_{KL}(P_0,\delta_0))$ converges to $0$. 
For any $n>0$,  choose a $P_{n}\in B_L(P_0,1/n)$ such that $P_{n}(t)$ is a step function, that is, the distribution function is a staircase function throughout the entire real line. Let $S_{n}:=\{x\in\mathcal{R}: P_{n}(x)-P_{n}(x-)>0\}$, $S_{n}$ is the set of all jump points of $P_{n}(t)$.
	
For the total variation case,  we have 
		\begin{eqnarray}
D\left(P_{n}||B_{TV}(P_0,\delta_0)\right) 	&=&\inf_{\{P\in B_{TV}(P_0,\delta_0)\}}D(P_{n}||P) \nonumber \\
		&\ge&\inf_{\{P\in B_{TV}(P_0,\delta_0)\}}P_{n}(S_{n})\log\frac{P_{n}(S_{n})}{P(S_{n})}+  P_{n}(S_{n}^c)\log\frac{P_{n}(S_{n}^c)}{P(S_{n}^c)} \nonumber \\
		&=&\inf_{\{P\in B_{TV}(P_0,\delta_0)\}}1\log\frac{1}{P(S_{n})}+0\log\frac{0}{P(S_{n}^c)} \nonumber \\
		&\ge&1\log\frac{1}{P_0(S_{n})+\delta_0}\nonumber \\
		&=&\log\frac{1}{\delta_0}, \nonumber
		\end{eqnarray}
where the first inequality is due to the data processing inequality of the KLD and the second inequality comes from the definition of $B_{TV}(P_0,\delta_0)$.		Therefore,
		\begin{eqnarray}
		\liminf_{n\to\infty}D(P_{n}||B_{TV}(P_0,\delta_0))&\ge&\log\frac{1}{\delta_0} \nonumber\\
		&>&D(P_0||B_{TV}(P_0,\delta_0))\nonumber\\
&=&0. \nonumber
		\end{eqnarray}

		Thus $D(P_{n}||B_{TV}(P_0,\delta_0))\nrightarrow D(P_0||B_{TV}(P_0,\delta_0))$ even though $P_{n}\wto P_0$. 

		As for the KLD case, for any $P$ such that $D(P||P_0)\le\delta_0$, $D(P_{n}||P)=\infty$ where $P_n$ is constructed in the same manner as above. Therefore, $D(P_{n}||B_{KL}(P_0,\delta_0))\nrightarrow D(P_0||B_{KL}(P_0,\delta_0))=0$.

	The assumption that $P_0(t)$ is continuous in $t$ is also necessary for the continuous property of the robust KLD to hold. We construct the following example to illustrate this point. Let $P_0$ be the distribution that $P_0(t)=0$ for $t<0$ and $P_0(t)=1$ for $t\ge0$, i.e., it is a degenerate random variable that equals to $0$ with probability $1$. Let $\mu_i$ be the distribution such that $\mu_i(t)=0$ for $t<0.5+\frac{1}{i}$ and $\mu_i(t)=1$ for $t\ge0.5+\frac{1}{i}$. Thus $\mu_i\wto\mu$ as $i\to\infty$, where $\mu(t)=0$ for $t<0.5$ and $\mu(t)=1$ for $t\ge0.5$. We can see that $D(\mu||B_L(P_0,0.5))=0$ since $\mu\in B_L(P_0,0.5)$. As $\mu_i\wto\mu$,
	\begin{eqnarray}
	\lim_{i\to\infty}D(\mu_i||B_L(P_0,0.5))&=&\lim_{i\to\infty}\log\frac{1}{0.5} \nonumber \\
	&>&D(\mu||B_L(P_0,0.5)), \nonumber 
	\end{eqnarray}
and the distribution in $B_L(P_0,0.5)$ achieving the KLD value of $\log 2$ is a degenerate one: it takes values of the two points $0.5$ and $0.5+1/i$ with equal probability.

 The proof of Theorem \ref{thm:continuous} sheds some light on the dynamics of the KLD of continuous distributions. Furthermore, the established continuity property of the robust KLD provides is key to solving the robust version of the universal hypothesis testing problem for the continuous case. This will be elaborated in the following section.

\section{Robust Universal Hypothesis Testing}

\subsection{Review of the large deviation approach}
While Hoeffding's test do not apply to distributions with continuous observations, Zeitouni and Gutman \cite{largedevi} developed a universal hypothesis test for distributions defined on the real line under a strictly weaker notion of optimality. Their approach relies on the large deviation theory, specifically, the general Sanov's theorem. For a given set $\Gamma\subset\mathcal{P}$, denote the closure and interior sets  of $\Gamma$ as $cl\Gamma$ and $int\Gamma$.
	\begin{thm}[General Sanov's Theorem]
		\label{thm:sanov}
		\cite{largebook} Given a probability set $\Gamma\subseteq\mathcal{P}$, for a probability measure $Q\notin\Gamma,$
		\begin{eqnarray}
		\inf_{P\in cl\Gamma}D(P||Q)&\le&\liminf_{n\rightarrow\infty}-\frac{1}{n}\log Q(\{x^n: \hat{\mu}_n\in\Gamma\}) \nonumber \\
		&\le&\limsup_{n\rightarrow\infty}-\frac{1}{n}\log Q(\{x^n: \hat{\mu}_n\in\Gamma\}) \nonumber\\
		&\le&\inf_{P\in int\Gamma}D(P||Q),\nonumber
		\end{eqnarray}
where $ \hat{\mu}_n$ is the empirical distribution defined in (\ref{eq:empirical}).
	\end{thm}
	
	The general Sanov's Theorem illustrates the large deviation principle for the empirical measures and is used extensively in the proof of Theorems \ref{thm:largedevi} and \ref{thm:robustuniversal}. For any set $\Gamma\subseteq\mathcal{P}$, define its $\delta-$smooth set to be $$\Gamma^{\delta}:=\cup_{P\in\Gamma}\{\mu\in\mathcal{P}:d_L(\mu, P)<\delta\}.$$
	The major contribution in \cite{largedevi} is summarized in the Theorem below.
	\begin{thm}
		\label{thm:largedevi}
		\cite{largedevi} Define $\Lambda$ as,
		\small
		\begin{eqnarray}
		\Lambda_1(n)=\Lambda_1:=\{\mu: D(B_L(\mu,2\delta)||P_0)\ge\eta\}^{\delta}, \ \ \Lambda_0:=\mathcal{P}\setminus \Lambda_1. \label{eq:largedevidetector}
		\end{eqnarray}
		\normalsize
		$\Lambda$ is $\delta-$optimal, i.e.,
		\begin{enumerate}
			\item
			$J^{P_0}(\Lambda)\ge\eta$.
			\item
			If $\Omega$ is a test such that $J^{P_0}(\Omega^{6\delta})\ge\eta$, then for any $P'\neq P_0$,
			\begin{eqnarray}
			I^{P'}(\Omega^{\delta})\le I^{P'}(\Lambda).
			\end{eqnarray}
		\end{enumerate}
	\end{thm}
	Theorem~\ref{thm:largedevi} applies to both discrete and $\mathcal{R}$-valued random variables. However, for the finite alphabet case, the corresponding detector as in (\ref{eq:largedevidetector}) yields weaker results than Hoeffding's detector \cite{Asymptotically}, a price paid for its generality. 
	
	With continuous alphabet, one has to be content with ``$\delta-$optimal" rather than ``optimal" for universal hypothesis testing if there is no restriction on the detector $\Omega$ while using the general Sanov's Theorem. For a test $\Omega$ defined by the partition of the probability space ($\Omega_1$, $\Omega_2$), either $\Omega_1$ or $\Omega_2$, may consist of only empirical distributions - it was established that the optimal test can depend on the observations only through the empirical distributions \cite{largedevi}. Suppose $\Omega_1$ consists of only empirical distributions. Then $int\Omega_1$ is empty and $cl\Omega_2$ equals to $\mathcal{P}$. It is also possible that the interior set and closure are too abstract or complicated to describe. In these cases, one can not take advantage of the general Sanov's Theorem to analyze the error exponents. That is why in Theorem \ref{thm:largedevi}, for an arbitrary test $\Omega$, we need to first perform $\delta-$smooth operation on it before comparing its error exponents to those of the test $\Lambda$. As such, without any restriction on the detector $\Omega$, one has to settle with ``$\delta-$optimality" instead of ``optimality" for the continuous case.
	
	Detector (\ref{eq:largedevidetector}) has a complicated form. Given the empirical distribution $\hat{\mu}_n$, it is difficult to determine whether $\hat{\mu}_n\in\Lambda_1$ or $\Lambda_2$ due to the following two reasons.
	
	\begin{itemize}
		\item
		Computing $D(B_L(\hat{\mu}_n,2\delta)||P_0)$ is an infinite dimension optimization problem. This is illustrated in Fig.~\ref{fig:BtoP} where one needs to find a continuous $\mu^*$ inside the shaded region such that $$D(\mu^*||P_0)=\inf_{\mu\in B_L(\hat{\mu}_n,2\delta)}D(\mu||P_0).$$
		\item
Suppose one can indeed evaluate $D(B_L(\hat{\mu}_n,2\delta)||P_0)$. If $D(B_L(\hat{\mu}_n,2\delta)||P_0)\ge\eta$ then $\hat{\mu}_n\in\Lambda_1$. However, if $D(B_L(\hat{\mu}_n,2\delta)||P_0)<\eta$, one needs to further check if $\hat{\mu}_n$ belongs to the $\delta$-smooth set of $\{\mu: D(B_L(\mu,2\delta)||P_0)\ge\eta\}$.
	\end{itemize}
	
	\begin{figure}[thb!]
		\vspace{-5pt}
		\begin{center}
			\includegraphics[width=3.6in]{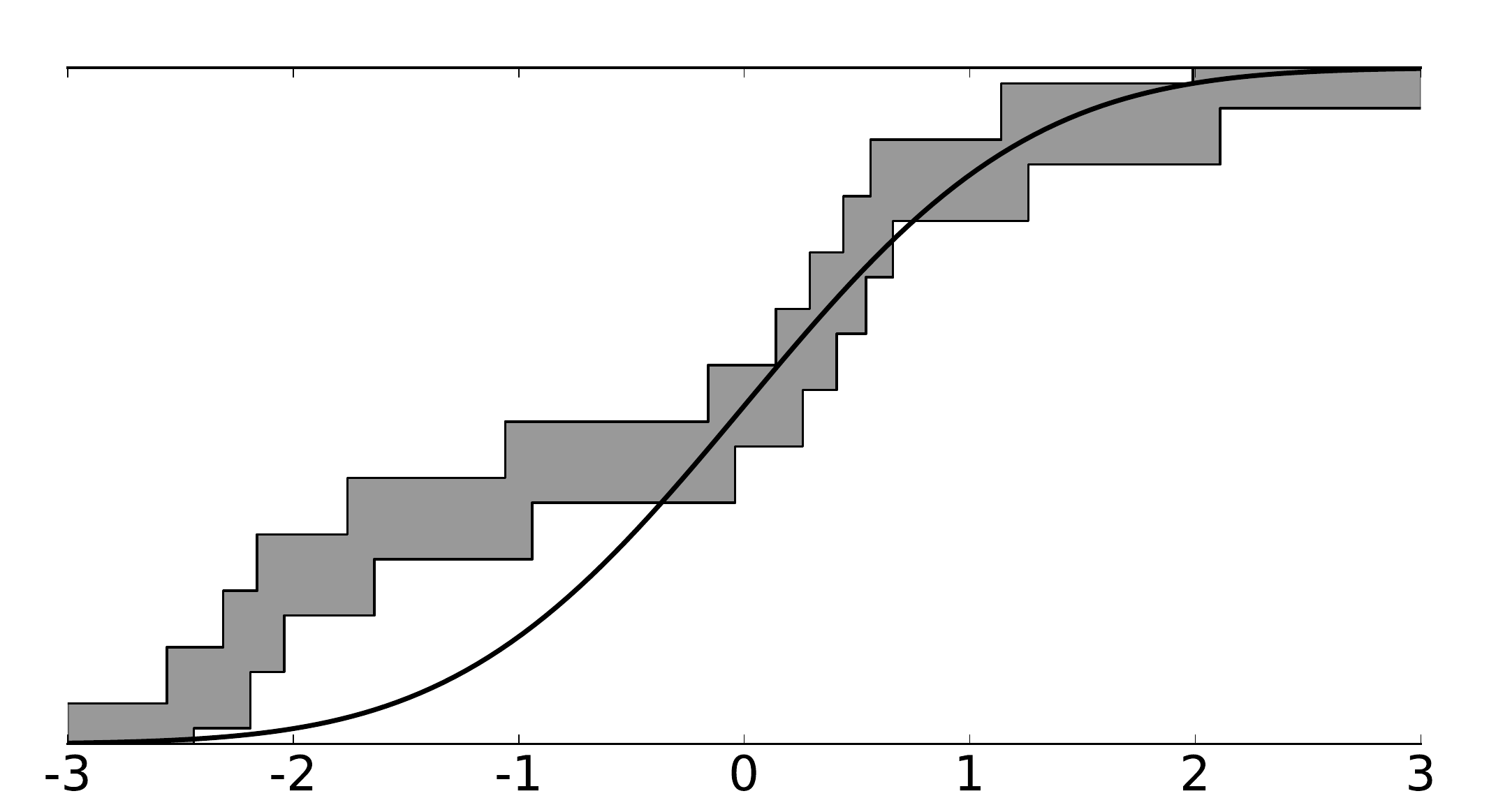}
			\caption[The shaded region is $B_L(\mu_n, 2\delta)$ and the solid line is $P_0$.]{The shaded region is $B_L(\hat{\mu}_n, 2\delta)$ and the solid line is $P_0$.}
			\label{fig:BtoP}
		\end{center}
	\end{figure}
	
	One of the difficulties to directly generalize the discrete case to the continuous case, as mentioned in \cite{largedevi}, is that the superlevel set $\{\mu\in\mathcal{P}: D(\mu||P_0)\geq\eta\}$ is not closed in $\mathcal{P}$, which we have discussed in detail in Section.\ref{subsec:continuityA}. In the next section, rather than ``$\delta-$smoothing'' the detector as one does in (\ref{eq:largedevidetector}), we generalize the hypothesis $\mathcal{H}_0$ from $P_0$ to $B_L(P_0,\epsilon_0)$. Then, we show that under the minimax criterion, the empirical likelihood ratio test is optimal.

\subsection{Robust Universal Hypothesis Testing}

\label{subsect:robustuniversal}

Let $\mathcal{P}_0:=B_L(P_0,\epsilon_0)$. The robust version of the universal hypothesis testing amounts to the testing of the following two hypotheses.
	\begin{eqnarray}
	\label{eq:robustUHT}
	\mathcal{H}_0: P\in\mathcal{P}_0, \quad \mathcal{H}_1: P=P'.
	\end{eqnarray}
	Here $P_0$ is assumed to be a known continuous distribution and $\epsilon_0>0$. $P'\notin\mathcal{P}_0$ and is unknown.
	
	Compared to the universal hypothesis test (\ref{eq:singleUHT}), the robust version (\ref{eq:robustUHT}) replaces a single $P_0$ with $B_L(P_0, \epsilon_0)$ for the null hypothesis. With this robust setting, the asymptotic NP criterion (\ref{eq:asmset}) is replaced by the following minimax asymptotic NP criterion: 
	\begin{eqnarray}
	\label{eq:robustasmset}
	\max_{\Omega}I^{P'}(\Omega) \ \ \text{\ s.t.\ } \ 
	J^{\mathcal{P}_0}(\Omega)\ge\eta,
	\end{eqnarray}
	where
	\begin{eqnarray} J^{\mathcal{P}_0}(\Omega):=\inf_{P\in\mathcal{P}_0}J^{P}(\Omega).
	\end{eqnarray}
Thus type II error exponent is maximized subject to a constraint on the worst type I error exponent. 	The reason that the L\'evy metric is used to define $\mathcal{P}_0$ is that the L\'evy metric is the weakest hence the most general one \cite{TV}. In another word, $B_L(P_0,\epsilon_0)$ contains all distributions that are close enough to $P_0$ as measured using any other metrics. An additional advantage is that the resulting optimal detector is rather intuitive and straightforward to implement. Theorem \ref{thm:robustuniversal} below describes the optimal solution to the robust universal hypothesis testing problem.
	
	\begin{thm}
		\label{thm:robustuniversal}
		For the robust universal hypothesis testing problem, the detector $\Lambda=\{\Lambda_0,\Lambda_1\}$ defined by, for some $\eta>0$,
		$$\Lambda_1(n)=\Lambda_1:=\{\mu: D(\mu||B_L(P_0,\epsilon_0))>\eta\}, \ \ \Lambda_0(n)=\mathcal{P}\setminus\Lambda_1,$$
satisfies the following properties: 
		\begin{enumerate}
			\item
			$J^{\mathcal{P}_0}(\Lambda)=\eta$.
			\item
			$I^{P'}(\Lambda)= D(\Lambda_0||P')$.
			\item
			For any detector $\Omega$ with $\Omega_1(n)=\Omega_1$ with $\Omega_1$ open, if
			\begin{eqnarray}
			\label{eq:contra1}
			J^{\mathcal{P}_0}(\Omega)>\eta,
			\end{eqnarray}
			then for any $P'\notin B_L(P_0,\epsilon_0)$,
			\begin{eqnarray}
			I^{P'}(\Omega)\le I^{P'}(\Lambda).
			\end{eqnarray}
		\end{enumerate}
	\end{thm}
	
Theorem~\ref{thm:robustuniversal} states that 
the detector
		\begin{eqnarray}
		\label{eq:universal_detector}
		D(\hat{\mu}_n||\mathcal{P}_0)\underset{{H_0}}{\overset{H_{1}}{\gtrless}}\eta,
		\end{eqnarray} 
is optimal in type II error decay rate among all detectors $\Omega=\{\Omega_0(n),\Omega_1(n)\}$ that have the same worst case type I error decay rate as $\Lambda$. In particular, the optimal type II error decay rate is 
precisely $ D(\Lambda_0||P')$ when $P'$ is the true distribution under $H_1$.
	\begin{proof}
	The three parts of the Theorem~\ref{thm:robustuniversal} are proved below.

1) From the general Sanov's theorem, we have  
	\begin{eqnarray}
	\inf_{P\in\mathcal{P}_0}J^{P}(\Lambda)&=&\inf_{P\in\mathcal{P}_0}\liminf_{n\rightarrow\infty}-\frac{1}{n}\log P^n(\{x^n: \hat{\mu}_n\in\Lambda_1\}) \nonumber \\ 
	&\ge&\inf_{P\in\mathcal{P}_0}\inf_{\mu\in cl\Lambda_1}D(\mu||P) \nonumber \\
	&=&\inf_{\mu\in cl\Lambda_1}D(\mu||\mathcal{P}_0) \nonumber\\
	&=&\eta, \nonumber
	\end{eqnarray}
	the last equality holds since $D(\mu||\mathcal{P}_0)$ is continuous in $\mu$ thus $cl\Lambda_1\subseteq\{\mu: D(\mu||\mathcal{P}_0)\ge\eta\}$. On the other hand,
	\begin{eqnarray}
	\inf_{P\in\mathcal{P}_0}J^{P}(\Lambda)&\le&\inf_{P\in\mathcal{P}_0}\limsup_{n\rightarrow\infty}-\frac{1}{n}\log P^n(\{x^n: \hat{\mu}_n\in\Lambda_1\}) \nonumber\\ 
	&\le&\inf_{P\in\mathcal{P}_0}\inf_{\mu\in int\Lambda_1}D(\mu||P) \nonumber\\
	&=&\eta. \label{eq:092304}
	\end{eqnarray}
	The last equality holds since $int\Lambda_1=\Lambda_1$.

2)	Again from the general Sanov's theorem, we have
	\begin{eqnarray}
	I^{P'}(\Lambda)&=&\liminf_{n\rightarrow\infty}-\frac{1}{n}\log P^{'n}(\{x^n: \hat{\mu}_n\in\Lambda^{}_0\}) \nonumber\\
	&\ge&\inf_{\mu\in cl\Lambda^{}_0}D(\mu||P') \nonumber\\
	&=& D(\Lambda_0||P'). \nonumber
	\end{eqnarray}
	The last equality holds since $cl\Lambda_0=\Lambda_0$.
	On the other hand, $\{\mu: D(\mu||\mathcal{P}_0)<\eta\}\subseteq int\Lambda_0$, thus,
	\begin{eqnarray}
	I^{P'}(\Lambda)&\le&\limsup_{n\rightarrow\infty}-\frac{1}{n}\log P^{'n}(\{x^n: \hat{\mu}_n\in\Lambda^{}_0\}) \nonumber\\
	&\le&\inf_{\mu\in int\Lambda^{}_0}D(\mu||P') \nonumber\\
	&\le&\inf_{\mu\in \{\mu: D(\mu||\mathcal{P}_1)<\eta\}} D(\mu||P')\nonumber \\
	&\le& D(\Lambda^{}_0||P'). \label{eq:p2closure}
	\end{eqnarray}	
	Inequality (\ref{eq:p2closure}) holds because of the following. There exists a distribution $P\in\mathcal{P}_0$ such that $D(P||P')<\infty$. For any $P_c\in\Lambda^{}_0$ and $0<\lambda<1$, we have $(1-\lambda)P_c+\lambda P\in\{\mu: D(\mu||\mathcal{P}_0)<\eta\}$ since 
	\begin{eqnarray}
D((1-\lambda)P_c+\lambda P||\mathcal{P}_0) 
	&\le&(1-\lambda)D(P_c||\mathcal{P}_0)+\lambda D(P||\mathcal{P}_0)) \label{eq:092301}\\
	&<&(1-\lambda)\eta+0 \label{eq:092302}\\
	&<&\eta, \nonumber
	\end{eqnarray}
	where (\ref{eq:092301}) comes from the fact that $D(\mu||\mathcal{P}_0)$ is convex in $\mu$ while inequality (\ref{eq:092302}) is due to the fact $P\in\mathcal{P}_0$. Thus, 
	\begin{eqnarray}
\inf_{\mu\in \{\mu: D(\mu||\mathcal{P}_1)<\eta\}} D(\mu||P') &\le&\lim_{\lambda\to0^+}D((1-\lambda)P_c+\lambda P||P') \nonumber\\
	&\le&\lim_{\lambda\to0^+}(1-\lambda)D(P_c||P')+\lambda D(P||P') \nonumber\\
	&\le&D(P_c||P'),\nonumber
	\end{eqnarray}
	the last inequality holds since $D(P||P')<\infty$. The above inequalities hold for any $P_c\in\Lambda_0$, thus we have 
	\begin{eqnarray}
	\inf_{\mu\in \{\mu: D(\mu||\mathcal{P}_0)<\eta\}} D(\mu||P') \le D(\Lambda^{}_1||P'). \nonumber
	\end{eqnarray}

3)	We have
	\begin{eqnarray}
 \inf_{P\in\mathcal{P}_0}D(\Omega_1||P) 
 &=&\inf_{P\in\mathcal{P}_0}D(int\Omega_1||P) \nonumber\\
	&\ge&\inf_{P\in\mathcal{P}_0}\liminf_{n\rightarrow\infty}-\frac{1}{n}\log P(\{x^n:  \hat{\mu}_n\in\Omega_1^{}\}) \nonumber\\
	&>&\eta.\nonumber
	\end{eqnarray}
	Therefore, $\Omega_1\subseteq\Lambda_1$, or equivalently, $\Lambda_0\subseteq\Omega_0$. Next, 
	\begin{eqnarray}
	I^{P'}(\Omega)&=&\liminf_{n\rightarrow\infty}-\frac{1}{n}\log P^{'n}(\{x^n: \hat{\mu}_n\in\Omega_0^{}\}) \nonumber\\
	&\le&\liminf_{n\rightarrow\infty}-\frac{1}{n}\log P^{'n}(\{x^n: \hat{\mu}_n\in \Lambda_0\})  \nonumber\\
	&=& I^{P'}(\Lambda).\nonumber
	\end{eqnarray}

	\end{proof}
	
	Compared to detector (\ref{eq:largedevidetector}) in Theorem \ref{thm:largedevi}, detector (\ref{eq:universal_detector}) has three main differences.
	\begin{itemize}
		\item
Computing $D(\hat{\mu}_n||B_L(P_0,\delta_0))$ is a finite dimension optimization problem, which is in essence finding a step function inside the shaded area that achieves the minimum KLD to $\hat{\mu}_n$ (see Fig.~\ref{fig:PtoB}). This can be shown to be a convex optimization problem with linear constraints thus can be readily solved via standard convex programs \cite{Yang-thesis}.
		\item
From Theorem~\ref{thm:largedevi}, the detector developed in \cite{largedevi} can not be compared directly to an arbitrary detector $\Omega$; instead,  $\Omega^{\delta}$ is used in establishing the optimality of the proposed detector. This ensures that $\Omega_1^{\delta}$ is open and $\Omega_0^{\delta}$ is closed yet this leads to a weaker sense of optimality, i.e., $\delta$-optimality.
		
		In Theorem \ref{thm:robustuniversal}, by restricting $\Omega_1$ to be independent of $n$ and assuming $\Omega_1$ is open, asymptotic NP optimality is established which is stronger than $\delta-$optimality.
		\item
		Theorem \ref{thm:largedevi} only provides the lower bound for error exponents, while Theorem \ref{thm:robustuniversal} characterizes the exact values of error exponents. Furthermore, while the error exponents $I$ and $J$ are defined using limit infimum, from the proof it can be seen that $I$ and $J$ remain unchanged if one uses limit to define error exponents. Therefore, Theorem \ref{thm:robustuniversal} gives an exact characterization of error exponents.
	\end{itemize}
	\begin{figure}[thb!]
		\vspace{-5pt}
		\begin{center}
			\includegraphics[width=3.6in]{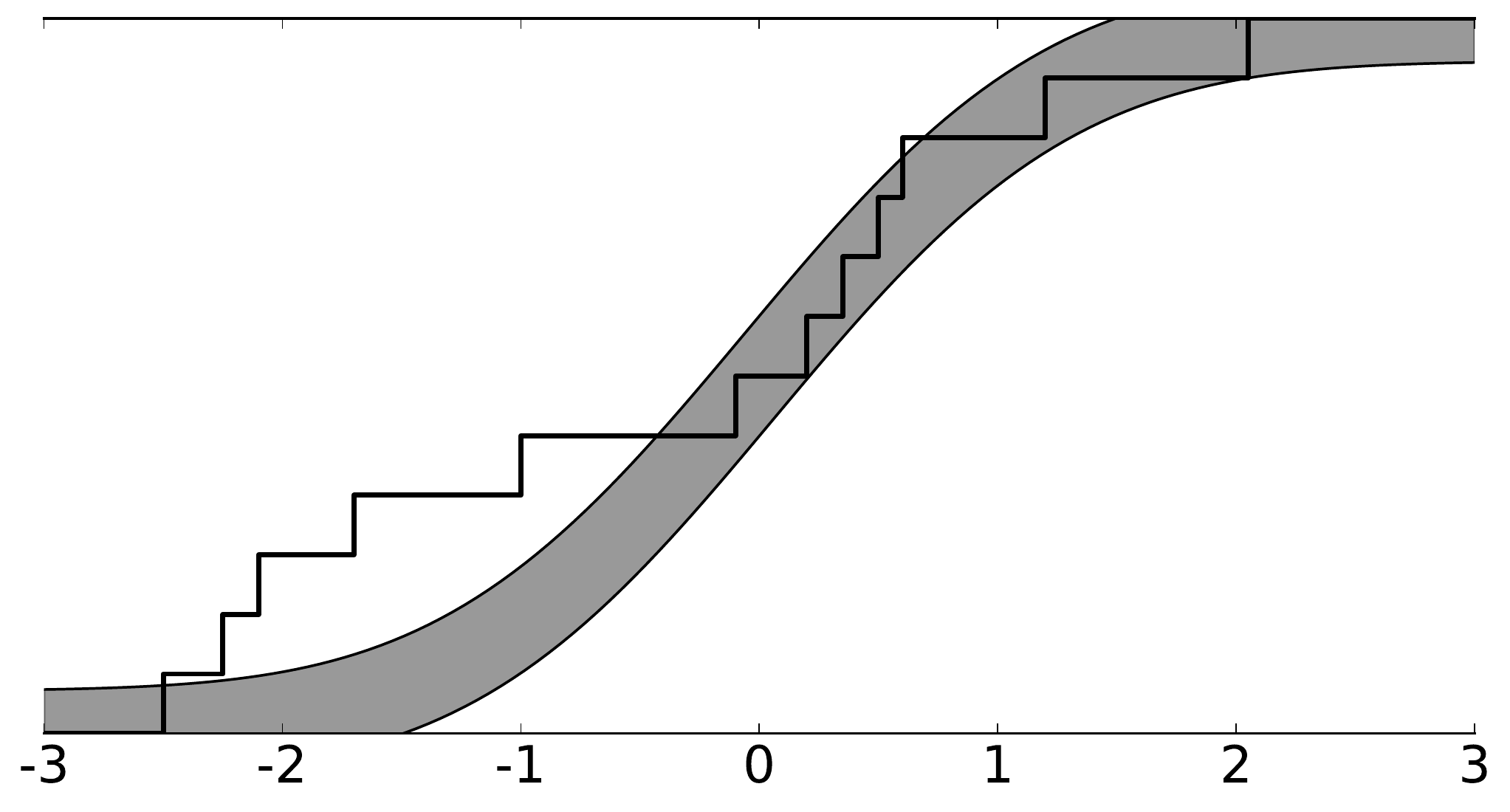}
			\caption[L\'evy ball of the standard normal distribution and the step function]{The shaded region is a L\'evy ball of the normal distribution and the step function is an example of $\hat{\mu}_n$.}
			\label{fig:PtoB}
		\end{center}
	\end{figure}
	
	Summarizing, by considering the universal hypothesis testing in the robust setting, the generalized empirical likelihood ratio test becomes optimal, and the construction of the detector and the proof of optimality are much simplified. 

	\section{Conclusion}
The Kullback-Leibler divergence (KLD) between a pair of distributions is only lower semicontinuous in the distribution functions for continuous observations. This is in contrast to the case with finite alphabet in which KLD is known to be continuous. As such, while simple and optimal solution may exist for some hypothesis testing problems involving finite alphabet observations, these results often do not generalize to the continuous case as the continuity of KLD plays a crucial role in obtaining the optimal test. 

The problem considered in the present paper is the universal hypothesis testing where the null hypothesis is specified by a nominal distribution whereas the alternative hypothesis is specified by a different but otherwise unknown distribution. With finite alphabet, Hoeffding's test, which is in essence a threshold test of the empirical KLD, is known to be asymptotically Neyman-Pearson (NP) optimal for the finite-alphabet case. For continuous observations, however, existing results have to resort to a weaker notion of optimality with a much more complicated detector compared with Hoeffding's detector.

This paper introduced the notion of the robust KLD, defined as the KLD between a distribution to the L\'evy ball of another distribution. In contrast to the classical KLD, this robust KLD was shown to be continuous in the first distribution function. Subsequently, by formulating a robust version of the universal hypothesis testing where the null hypothesis is specified by a L\'evy ball centered around the nominal distribution, it was established that the generalized empirical likelihood ratio test is optimal under the asymptotic minimax NP criterion whose error exponents were characterized precisely. Additionally, the test itself is also much more intuitive and easier to evaluate compared with existing approaches in the literature. 
	
	\appendices
	\section{Proof of Theorem \ref{thm:continuous}} 
	\label{app:continuous}
	The set of all partitions $\mathcal{A}=(A_1, \cdots, A_{|\mathcal{A}|})$ of $\mathcal{R}$ into a finite number of sets $A_i$ is denoted by $\Pi$. For a given partition $\mathcal{A}$, denote by $P^{\mathcal{A}}$ the quantized (discrete) probability over $\mathcal{A}$ of a probability distribution $P\in\mathcal{P}$. Thus $P^{\mathcal{A}}$ is a $|\mathcal{A}|$ dimensional vector $(P(A_1),P(A_2),\cdots, P(A_{|\mathcal{A}|}))\in\mathcal{R}^{|\mathcal{A}|}$. We introduce a new definition of KLD utilizing partitions, which is equivalent to the classical definition using the Radon-Nikodym derivative in (\ref{def:kl2}).
	\begin{Def}[The KLD \cite{kld_partition}]
		The KLD between $P\in\mathcal{P}$ and $Q\in\mathcal{P}$ is defined as, 
		\begin{eqnarray}
		\label{eq:kl}
		D(P||Q)=\sup_{\mathcal{A}\in\Pi}D^{}(P^{\mathcal{A}}|| Q^{\mathcal{A}}),
		\end{eqnarray}
		where 
		\begin{eqnarray}
		D^{}(P^{\mathcal{A}}||Q^{\mathcal{A}})=\sum_{i=1}^{|\mathcal{A}|}P(A_i)\log\frac{P(A_i)}{Q(A_i)}. \nonumber
		\end{eqnarray}
	\end{Def}
	
	The following lemma generalizes (\ref{eq:kl}) from the classical KLD to the robust KLD. For set $\Gamma\subseteq\mathcal{P}$, we define $\Gamma^{\mathcal{A}}:=\{P^{\mathcal{A}}: P\in\Gamma\}.$
	\begin{lem}
		\label{lem:pequalm}
		For $\mu, P_0\in\mathcal{P}$ and $\delta_0>0$, $D(\mu||B_L(P_0,\delta_0)) = \sup_{\mathcal{A}\in\Pi}D(\mu^{\mathcal{A}}||B_L^{\mathcal{A}}(P_0,\delta_0)).$
	\end{lem}
	\begin{proof}
		The probability space $\mathcal{P}$ defined on $(\mathcal{R},\mathcal{F})$, where $\mathcal{R}$ is the real line and $\mathcal{F}$ is the sigma-algebra that contains all the Borel sets of $\mathcal{R}$,  is not compact with respect to the weak convergence \cite{spaceproperty}. Let $\mathcal{M}$ denote the space of {\it finitely} additive and non-negative set functions on $(\mathcal{R},\mathcal{F})$ with $M(\mathcal{R})=1$ for $M\in\mathcal{M}$. Thus, we relax the countable additivity to finite additivity. As a consequence, $\mathcal{P}\subseteq\mathcal{M}$ and $\mathcal{M}$ is compact with respect to the weak convergence  \cite{simple}. Similarly, we define $M(t):=M((-\infty,t])$ and $M(t)$ is a right continuous non-decreasing function on $\mathcal{R}$. As with $P(t)$ and $P\in\mathcal{P}$, $M(t)$ and $M\in \mathcal{M}$ are equivalent since one is uniquely determined by the other.
		
		$\mathcal{P}$ is equivalent to the set of right continuous non-decreasing functions on $\mathcal{R}$ with $P(-\infty)=0$ and $P(\infty)=1$ if $P\in\mathcal{P}$ \cite{spaceproperty}, while $M$ is equivalent to the set of right continuous non-decreasing functions on $\mathcal{R}$ with $M(-\infty)\geq0$ and $M(\infty)\leq1$ if $M\in\mathcal{M}$. The L\'evy metric $d_L$ and the KLD extend unchanged to $F\in\mathcal{M}$ and $G\in\mathcal{M}$ \cite{spaceproperty, simple}. The following three steps constitute the proof of the lemma,
		\begin{eqnarray}
		D(\mu||B_L(P_0,\delta_0))&=& D(\mu||\bar{B}_L(P_0,\delta_0)),\ \ \ \   
		\label{eq:exchange1}\\
		D(\mu||\bar{B}_L(P_0,\delta_0)) &=& \sup_{\mathcal{A}\in\Pi}D(\mu^{\mathcal{A}}||\bar{B}_L^{\mathcal{A}}(P_0,\delta_0)), 
		\label{eq:exchange2} \ \ \ \  \\
		\sup_{\mathcal{A}\in\Pi}D(\mu^{\mathcal{A}}||\bar{B}_L^{\mathcal{A}}(P_0,\delta_0))&=&\sup_{\mathcal{A}\in\Pi}D(\mu^{\mathcal{A}}||B_L^{\mathcal{A}}(P_0,\delta_0)), \label{eq:exchange3} \ \ \ \ 
		\end{eqnarray}
		where $\bar{B}_L(P_0,\delta_0):=\{P\in\mathcal{M}: d_L(P,P_0)\le\delta_0\}$.
		We now prove (\ref{eq:exchange1})-(\ref{eq:exchange3}). 
We first prove (\ref{eq:exchange1}). Note that $\bar{B}_L(P_0,\delta_0)$ is closed with respect to the weak convergence, thus is compact since $\mathcal{M}$ is compact. Let 
			$$P_{\mu}:=\arg\inf_{\{P\in \bar{B}_L(P_0,\delta_0)\}}D(\mu||P),$$ 
			the existence of $P_{\mu}$ is guaranteed since $D(\mu||P)$ is lower semicontinuous and lower semicontinuous function attains its infimum on a compact set. Assume $P_{\mu}\in\mathcal{M}\setminus\mathcal{P}$, then there exists a $\delta>0$ such that $P_{\mu}(-\infty)\geq\delta$ or $P_{\mu}(+\infty)\leq1-\delta$. Without loss of generality, we assume $P_{\mu}(-\infty)=\delta$ and $P_{\mu}(+\infty)=1$; other cases can be proved in a similar manner. Let $s$ denote the minimum $t$ such that $P_0(t-\delta_0)\geq\delta_0$, we construct $P'_{\mu}$ as follows.
			\begin{itemize}
				\item
				If $P_{\mu}(s)=\delta$, let
				\[ P'_{\mu}(t) =
				\begin{cases}
				0       & \quad \text{if } t<s,\\
				P_{\mu}(t)  & \quad \text{if } t\ge s.\\
				\end{cases}
				\]
				Since $\inf_tP'_{\mu}(t)=0$ and $\sup_tP'_{\mu}(t)=1$, $P'_{\mu}\in\mathcal{P}$. In addition, it can be easily verified that $d_L(P'_{\mu},P_0)\le\delta_0$. Therefore, $P'_{\mu}\in B_L(P_0,\delta_0)$ and $D(\mu||P'_{\mu})=D(\mu||P_{\mu})$.
				\item
				If $P_{\mu}(s)>\delta$, let
				\[ P'_{\mu}(t) =
				\begin{cases}
				\frac{(P_{\mu}(t)-\delta)P_{\mu}(s)}{P_{\mu}(s)-\delta}      & \quad \text{if } t<s,\\
				P_{\mu}(t)  & \quad \text{if } t\ge s.\\
				\end{cases}
				\]
				$\inf_tP'_{\mu}(t)=0$ and $\sup_tP'_{\mu}(t)=1$ thus $P'_{\mu}\in\mathcal{P}$. For $t<s$, $$\frac{(P_{\mu}(t)-\delta)P_{\mu}(s)}{P_{\mu}(s)-\delta}\le P_{\mu}(t)\Leftrightarrow P_{\mu}(t)\le P_{\mu}(s),$$ 
				then $d_L(P'_{\mu},P_0)\le\delta_0$ because
				$$P_0(t-\delta_0)-\delta_0<0\le P'_{\mu}(t)\le P_{\mu}(t)\le P_0(t+\delta_0)+\delta_0.$$
				Therefore, we have $P'_{\mu}\in B_L(P_0,\delta_0)$. Also $P'_{\mu}$ achieves the infimum since,  
				\begin{eqnarray}
 D(\mu||P'_{\mu}) 
 &=&\int_{-\infty}^{s-}d\mu(t)\log\frac{d\mu(t)}{dP'_{\mu}(t)}+\int_{s}^{\infty}d\mu(t)\log\frac{d\mu(t)}{dP'_{\mu}(t)} \nonumber \\
				&=&\mu(s-)\log\frac{P_{\mu}(s)-\delta}{P_{\mu}(s)}+\int_{-\infty}^{s-}d\mu(t)\log\frac{d\mu(t)}{dP_{\mu}(t)} \nonumber \\
				&&+\int_{s}^{\infty}d\mu(t)\log\frac{d\mu(t)}{dP_{\mu}(t)} \nonumber \\
				&=&\mu(s-)\log\frac{P_{\mu}(s)-\delta}{P_{\mu}(s)}+D(\mu||P_{\mu}) \nonumber \\
				&\le& D(\mu||P_{\mu}). \nonumber
				\end{eqnarray}
			\end{itemize}
			Therefore in either case, there exists $P'_{\mu}\in B_L(P_0,\delta_0)$ such that
			\begin{eqnarray}
			\label{eq:achieves_infimum}
			P'_{\mu}=\arg\inf_{\{P\in \bar{B}_L(P_0,\delta_0)\}}D(\mu||P). 
			\end{eqnarray}

To prove (\ref{eq:exchange2}), note that Lemma 2.4 in \cite{GOR} shows that $$D(B_L(P_0,\delta_0)||\mu)=\sup_{\mathcal{A}\in\Pi}D(B_L^{\mathcal{A}}(P_0,\delta_0)||\mu^{\mathcal{A}}).$$ 
Using a parallel proof, one can establish that $$D(\mu||\bar{B}_L(P_0,\delta_0))=\sup_{\mathcal{A}\in\Pi}D(\mu^{\mathcal{A}}||\bar{B}_L^{\mathcal{A}}(P_0,\delta_0)),$$ i.e., (\ref{eq:exchange2}) holds.
			
Finally, (\ref{eq:exchange3}) holds since for any $\mathcal{A}\in\Pi$, $\bar{B}_L^{\mathcal{A}}(P_0,\delta_0)=B_L^{\mathcal{A}}(P_0,\delta_0)$. 
	\end{proof}
	
	The robust KLD is continuous in the radius of the L\'evy ball under a mild assumption, as stated in the lemma below.
	\begin{lem}
		\label{lem:kld_left_conti}
		Given $\mu, P_0\in\mathcal{P}$ and $\delta_0>0$, if $P_0(t)$ is continuous in $t$, then $D(\mu||B_L(P_0,\delta_0))$ is continuous in $\delta_0$.
	\end{lem}
	\begin{proof}
		Let $\delta\in(0,\delta_0)$. $D(\mu||B_L(P_0,\delta_0))$ is left continuity in $\delta_0$ if
		$D(\mu||\{P\in\mathcal{P}: d_L(P,P_0)<\delta_0\})= D(\mu||B_L(P_0,\delta_0))$.		It is easy to see that 
$$D(\mu||\{P\in\mathcal{P}: d_L(P,P_0)<\delta_0\})\geq D(\mu||B_L(P_0,\delta_0)).$$ 
Thus we only need to show the other direction.
		Denote 
		$$P_{\delta}=\arginf_{\{P\in B_L(P_0,\delta)\}}D(\mu||P),$$ $$P_{\delta_0}=\arginf_{\{P\in B_L(P_0,\delta_0)\}}D(\mu||P).$$
		The existence of $P_{\delta}$ and $P_{\delta_0}$ is guaranteed by (\ref{eq:achieves_infimum}). For any $0<\lambda<1$, $$d_L(\lambda P_{\delta}+(1-\lambda)P_{\delta_0},P_0)<\delta_0,$$
		which may not hold if $P_0(t)$ is not continuous in $t$. Then,
		\begin{eqnarray}
D(\mu||\{P\in\mathcal{P}: d_L(P,P_0)<\delta_0\}) 
	&\le&\lim_{\lambda\to0+}
		D(\mu||\lambda P_{\delta}+(1-\lambda)P_{\delta_0}) \nonumber \\
		&\le&\lim_{\lambda\to0+}\lambda D(\mu|| P_{\delta})+(1-\lambda)D(\mu|| P_{\delta_0}) \nonumber \\
		&=&D(\mu||P_{\delta_0}) \nonumber\\
		&=&D(\mu||B_L(P_0,\delta_0)). \nonumber
		\end{eqnarray}
		Therefore $D(\mu||B_L(P_0,\delta_0))$ is left continuous in $\delta_0$. 
		
		The rest is to show $D(\mu||B_L(P_0,\delta_0))$ is right continuous in $\delta_0$. Since $D(\mu||B_L(P_0,\delta_0))$ is decreasing in $\delta_0$, we only need to show:
		\begin{eqnarray}
		\lim_{n\to\infty}D\left(\mu||B_L(P_0,\delta_0+\frac{1}{n})\right)\ge D(\mu||B_L(P_0,\delta_0)).
		\end{eqnarray}
		From (\ref{eq:achieves_infimum}), there exists $P_n\in B_L(P_0,\delta_0+\frac{1}{n})$ such that $D(\mu||P_n)=D\left(\mu||B_L(P_0,\delta_0+\frac{1}{n})\right)$. $\mathcal{M}$ is compact, $P_n$ converges to $P^*\in\mathcal{M}$. Since $P^*\in\ \bar{B}_L(P_0,\delta_0+\frac{1}{n})$ for any $n$, $P^*\in\bar{B}_L(P_0,\delta_0)$. We have,
		\begin{eqnarray}
\lim_{n\to\infty}D\left(\mu||B_L(P_0,\delta_0+\frac{1}{n})\right)&=&\lim_{n\to\infty}D(\mu||P_n) \nonumber\\
		&\ge&D(\mu||P^*)  \label{eq:091701} \nonumber\\
		&\ge&D(\mu||\bar{B}_L(P_0,\delta_0)) \nonumber\\
		&=&D(\mu||B_L(P_0,\delta_0)),\nonumber
		\end{eqnarray}
		where (\ref{eq:091701}) comes from the fact that the KLD is lower semicontinuous and the last equality was proved in (\ref{eq:exchange1}). Therefore $D(\mu||B_L(P_0,\delta_0))$ is right continuous in $\delta_0$.
	\end{proof}

	\begin{lem}
		\label{lem:convex}
		For $\mu, P_0\in\mathcal{P}$ and $\delta_0>0$, $D(\mu||B_L(P_0,\delta_0))$ is a convex function of $\mu.$ In addition, for any partition $\mathcal{A}$ of the real line, $D(\mu^\mathcal{A}||B_L^\mathcal{A}(P_0,\delta_0))$ is convex in $\mu^\mathcal{A}.$
	\end{lem}
	\begin{proof}
		Let $P_i=\arg\inf_{\{P\in B_L(P_0,\delta_0)\}}D(\mu_i||P)$, $i=1, 2$. For any $0<\lambda<1$, $\lambda P_1+(1-\lambda)P_2\in B_L(P_0,\delta_0)$, thus,
		\begin{eqnarray}
D(\lambda\mu_1+(1-\lambda)\mu_2||B_L(P_0,\delta_0))
		&\le& D(\lambda\mu_1+(1-\lambda)\mu_2||\lambda P_1+(1-\lambda)P_2) \nonumber \\
		&\le& \lambda D(\mu_1|| P_1)+(1-\lambda)D(\mu_2||P_2) \nonumber\\
		&=&\lambda D(\mu_1||B_L(P_0,\delta_0))+(1-\lambda)D(\mu_2||B_L(P_0,\delta_0)). \nonumber
		\end{eqnarray}
		Therefore, $D(\mu||B_L(P_0,\delta_0))$ is a convex function of $\mu$. That $D(\mu^\mathcal{A}||B_L^\mathcal{A}(P_0,\delta_0))$ is convex in $\mu^\mathcal{A}$ follows a similar argument.
	\end{proof}
	
	\begin{lem}
		\label{lem:suptor}
		Given $\mu_0, P_0\in\mathcal{P}$ and $\delta, \delta_0>0$, we have $$\sup_{\mu\in B_L(\mu_0,\delta)}D(\mu||B_L(P_0,\delta_0))=\sup_{x\in\mathcal{R}}D(\mu_x^{\delta}||B_L(P_0,\delta_0)),$$
		where 
		\begin{eqnarray}
		\mu_x^{\delta}(t)=
		\begin{cases}
		\max(0,\mu_0(t-\delta)-\delta))       & \quad \text{if } t<x,\\
		\min(1,\mu_0(t+\delta)+\delta))  & \quad \text{if } t\ge x.
		\end{cases}\nonumber
		\end{eqnarray}
	\end{lem}
	\begin{proof}
		We have
		\begin{eqnarray}
\sup_{\mu\in B_L(\mu_0,\delta)}D(\mu||B_L(P_0,\delta_0)) 
		&=&\sup_{\mu\in B_L(\mu_0,\delta)}\sup_{\mathcal{A}\in\Pi}D(\mu^{\mathcal{A}}||B_L^{\mathcal{A}}(P_0,\delta_0)) \label{eq:usepequalm}\\
		&=&\sup_{\mathcal{A}\in\Pi}\sup_{\mu\in B_L(\mu_0,\delta)}D(\mu^{\mathcal{A}}||B_L^{\mathcal{A}}(P_0,\delta_0))\nonumber \\
		&=&\sup_{\mathcal{A}\in\Pi}\sup_{\mu^{\mathcal{A}}\in B_L^{\mathcal{A}}(\mu_0,\delta)}D(\mu^{\mathcal{A}}||B_L^{\mathcal{A}}(P_0,\delta_0)). \label{eq:102001}
		\end{eqnarray} 
		Equality (\ref{eq:usepequalm}) comes from Lemma \ref{lem:pequalm}. 
		
		Fix a finite partition $\mathcal{A}$ of the real line. Without loss of generality we can assume $|\mathcal{A}|=n$ and
$$\mathcal{A}=\{(-\infty,a_1],(a_1,a_2],\cdots,(a_{n-2},a_{n-1}],(a_{n-1},\infty)\}.$$ 
The partition $\mathcal{A}$ over the probability space $\mathcal{P}$ can be represented as an $n$-dimensional polytope. Denote the $n-$dimensional point $\mathbf{x}=(x_1,x_2,\cdots,x_n)$,
		\begin{eqnarray}
		\mathcal{P}^{\mathcal{A}}=\{\mathbf{x}\in\mathcal{R}^n: \sum_ix_i=1 \text{ and } \forall i, 0\le x_i\le 1\}. \nonumber
		\end{eqnarray}
		Similarly, the partition $\mathcal{A}$ over the set $B_L(\mu_0,\delta)$ is also an $n$-dimensional polytope inside $\mathcal{P}^{\mathcal{A}}$, 
		\begin{eqnarray}
		B_L^{\mathcal{A}}(\mu_0,\delta)=\{\mathbf{x}\in\mathcal{P}^{\mathcal{A}}: \forall 1\le j\le n-1, L_j^{}\le\sum_{i=1}^jx_i\le U_j^{} \}, \nonumber
		\end{eqnarray}
		where $L_j^{}=\max(0,\mu_0(a_j-\delta)-\delta), U_j^{}=\min(1,\mu_0(a_j+\delta)+\delta)$. We can assume for any $1\le j\le n-2$, $U_{j}^{}>L_{j+1}^{}$, otherwise we can make $\mathcal{A}$ finer such that the new partition (denoted as $\mathcal{A}$ again) has the property that $a_{j+1}\le a_{j}+\delta$ for $1\le j\le n-2$. It can be verified that for each $1\le j\le n-2$, $U_{j}^{}>L_{j+1}^{}$. The reason that such an $\mathcal{A}$ can be finite is that $\mu_0(t)$ is a bounded non-decreasing function. 
		
		A point $\mathbf{x}$ is a vertex of $ B_L^{\mathcal{A}}(\mu_0,\delta)$ if and only if $\sum_{i=1}^jx_i$ equals $L_j^{}$ or $U_j^{}$ for any $1\le j\le n-1$, $\sum_{i=1}^nx_i=1$ and $0\le x_i\le1$. Since $U_{j}^{}>L_{j+1}^{}$ for any $1\le j\le n-2$, for a vertex $\mathbf{x}$, once $\sum_{i=1}^jx_i=U_j^{}$ for some $j$, then for any $k>j$ we have $\sum_{i=1}^k x_i=U_k^{}$. 
		
		Therefore there are $n$ vertices $\mathbf{x}^1,\cdots,\mathbf{x}^n$ of $B_L^{\mathcal{A}}(\mu_0,\delta)$ that satisfy the property that $\sum_{i=1}^{j}x^k_i=L_j^{}$ for $j<k$, $\sum_{i=1}^{j}x^k_i=U_j^{}$ for $j\ge k$. Or equivalently, if we denote $L_0^{}=0$ and $U_n^{}=1$, for $1\le k\le n$,
		\[ x_i^k =
		\begin{cases}
		L_i^{}-L_{i-1}^{}       & \quad \text{if } i<k,\\
		U_i^{}-L_{i-1}^{}  & \quad \text{if } i=k,\\
		U_i^{}-U_{i-1}^{}  & \quad \text{if } i>k.
		\end{cases}
		\]
		From Lemma \ref{lem:convex}, $D(\cdot||B_L^{\mathcal{A}}(P_0,\delta_0))$ is a convex function, thus the supremum on the polytope $B_L^{\mathcal{A}}(\mu_0,\delta)$ is achieved at its vertices. Let 
		\begin{eqnarray}
		\mu_x^{\delta}(t)=
		\begin{cases}
		\max(0,\mu_0(t-\delta)-\delta))       & \quad \text{if } t<x,\\
		\min(1,\mu_0(t+\delta)+\delta))  & \quad \text{if } t\ge x.
		\end{cases}\nonumber
		\end{eqnarray}
		Then any $\mathbf{x}^k$ is a quantization of $\mu_x^{\delta}$ over the partition $\mathcal{A}$ for some $x$.
		\begin{eqnarray}
\sup_{\mu^{\mathcal{A}}\in B_L^{\mathcal{A}}(\mu_0,\delta)}D(\mu^{\mathcal{A}}||B_L^{\mathcal{A}}(P_0,\delta_0)) 
		&=&
		\max_{k}D(\mathbf{x}^k||B_L^{\mathcal{A}}(P_0,\delta_0))  \nonumber \\
		&\leq&
		\sup_{x}D((\mu_x^{\delta})^{\mathcal{A}}||B_L^{\mathcal{A}}(P_0,\delta_0)) \nonumber  \\
		&\leq&\sup_{\mathcal{A}\in\Pi}\sup_{x\in\mathcal{R}}D((\mu_x^{\delta})^{\mathcal{A}}||B_L^{\mathcal{A}}(P_0,\delta_0)),  \nonumber \\
		&=&\sup_{x\in\mathcal{R}}\sup_{\mathcal{A}\in\Pi}D((\mu_x^{\delta})^{\mathcal{A}}||B_L^{\mathcal{A}}(P_0,\delta_0)),  \nonumber \\
		&=&\sup_{x\in\mathcal{R}}D(\mu_x^{\delta}||B_L(P_0,\delta_0)), \label{eq:102002}
		\end{eqnarray}	
		the last equality comes from Lemma \ref{lem:pequalm}. From (\ref{eq:102001}) and (\ref{eq:102002}), we have
		$$
		\sup_{\mu\in B_L(\mu_0,\delta)}D(\mu||B_L(P_0,\delta_0))\le\sup_{x\in\mathcal{R}}D(\mu_x^{\delta}||B_L(P_0,\delta_0)).
		$$
		For the other direction, since $\mu_x^{\delta}\in B_L(\mu_0,\delta)$, 
		$$
		\sup_{\mu\in B_L(\mu_0,\delta)}D(\mu||B_L(P_0,\delta_0))\ge\sup_{x\in\mathcal{R}}D(\mu_x^{\delta}||B_L(P_0,\delta_0)).
		$$
		Therefore, 
		$$
		\sup_{\mu\in B_L(\mu_0,\delta)}D(\mu||B_L(P_0,\delta_0))=\sup_{x\in\mathcal{R}}D(\mu_x^{\delta}||B_L(P_0,\delta_0)).
		$$
	\end{proof}
	
	A direct result of the above lemma is the boundedness of the robust KLD, which is stated below.
	\begin{lem}
		\label{pro:finite}
		$\sup_{\mu, P_0\in\mathcal{P}}D(\mu||B_L(P_0,\delta_0)=\log\frac{1}{\delta_0}$.
	\end{lem}
	
	\begin{proof}
		We construct a distribution $S_0\in\mathcal{P}$ such that $S_0(t)=0$ for $t<0$, and $S_0(t)=1$ for $t\ge 0$, then $\mathcal{P}=B_L(S_0,1)$ since $d_L$ is bounded by $1$. According to Lemma \ref{lem:suptor}, 
		\begin{eqnarray}
		\sup_{\mu\in B_L(S_0,1)}D(\mu||B_L(P_0,\delta_0))=\sup_{x\in\mathcal{R}}D(\mu_x^{\delta}||B_L(P_0,\delta_0)), \nonumber
		\end{eqnarray}
		where $\mu_x^{1}(t)=0$ for $t<x$, and $\mu_x^{1}(t)=1$ for $t\ge x$. Therefore,
		\begin{align}
		&\ \sup_{x\in\mathcal{R}}D(\mu_x^{\delta}||B_L(P_0,\delta_0)) \nonumber\\
		={}&\ \sup_{x\in\mathcal{R}}\log\frac{1}{\min(1,P_0(x+\delta_0)+\delta_0)-\max(0,P_0(x-\delta_0)-\delta_0)}\nonumber \\
		={}&\ \log\frac{1}{\delta_0}, \nonumber
		\end{align}
		the last equality comes from the fact that $$\min(1,P_0(t+\delta_0)+\delta_0)-\max(0,P_0(t-\delta_0)-\delta_0)\ge\delta_0$$ 
		and 
		$$\lim_{x\to\infty}\min(1,P_0(x+\delta_0)+\delta_0)-\max(0,P_0(x-\delta_0)-\delta_0)=\delta_0,$$
		which means a finitely additive measure that belongs to $\mathcal{M}\setminus\mathcal{P}$ can always achieve the supremum for any $P_0$.
	\end{proof}
	
	With all the previous lemmas, we now show that robust KLD is upper semicontinuous.
	\begin{lem}
		\label{lem:uppersemi}
		Given $P_0\in\mathcal{P}$ and $\delta_0>0$, if $P_0(t)$ is continuous in $t$, then $D(\mu||B_L(P_0,\delta_0))$ is upper semicontinuous in $\mu$ with respect to the weak convergence.
	\end{lem}
	\begin{proof}
		For any fixed $\mu_0\in\mathcal{P}$, the statement is equivalent to proving that when $\delta\to0$,
		\begin{eqnarray}
		\lim_{\delta\to0}\sup_{\mu\in B_L(\mu_0,\delta)}D(\mu||B_L(P_0,\delta_0))\le D(\mu_0||B_L(P_0,\delta_0)).
		\end{eqnarray}
		From Lemma \ref{lem:suptor}, it is equivalent to proving
		$$\lim_{\delta\to0}\sup_{x\in\mathcal{R}}D(\mu_x^{\delta}||B_L(P_0,\delta_0)) \le D(\mu_0||B_L(P_0,\delta_0)).$$ 
		Denote $u_{-\delta}$ as the left boundary of support set of distribution $\mu(t+\delta)$ and $u_{-\delta}^{\delta}$ as the infimum $x$ such that $\mu(x+\delta)=1-\delta$. 
Similarly, denote $u_{\delta}$ as the left boundary of distribution $\mu(t-\delta))$ and $u_{\delta}^{\delta}$ as the infimum $x$ such that $\mu(x-\delta)=1$. 
 Fig.~\ref{fig:udelta} illustrates these locations.  Note that these boundary points may not be finite.
		\begin{figure}[thb!]
			\begin{center}
				\includegraphics[width=3.8in]{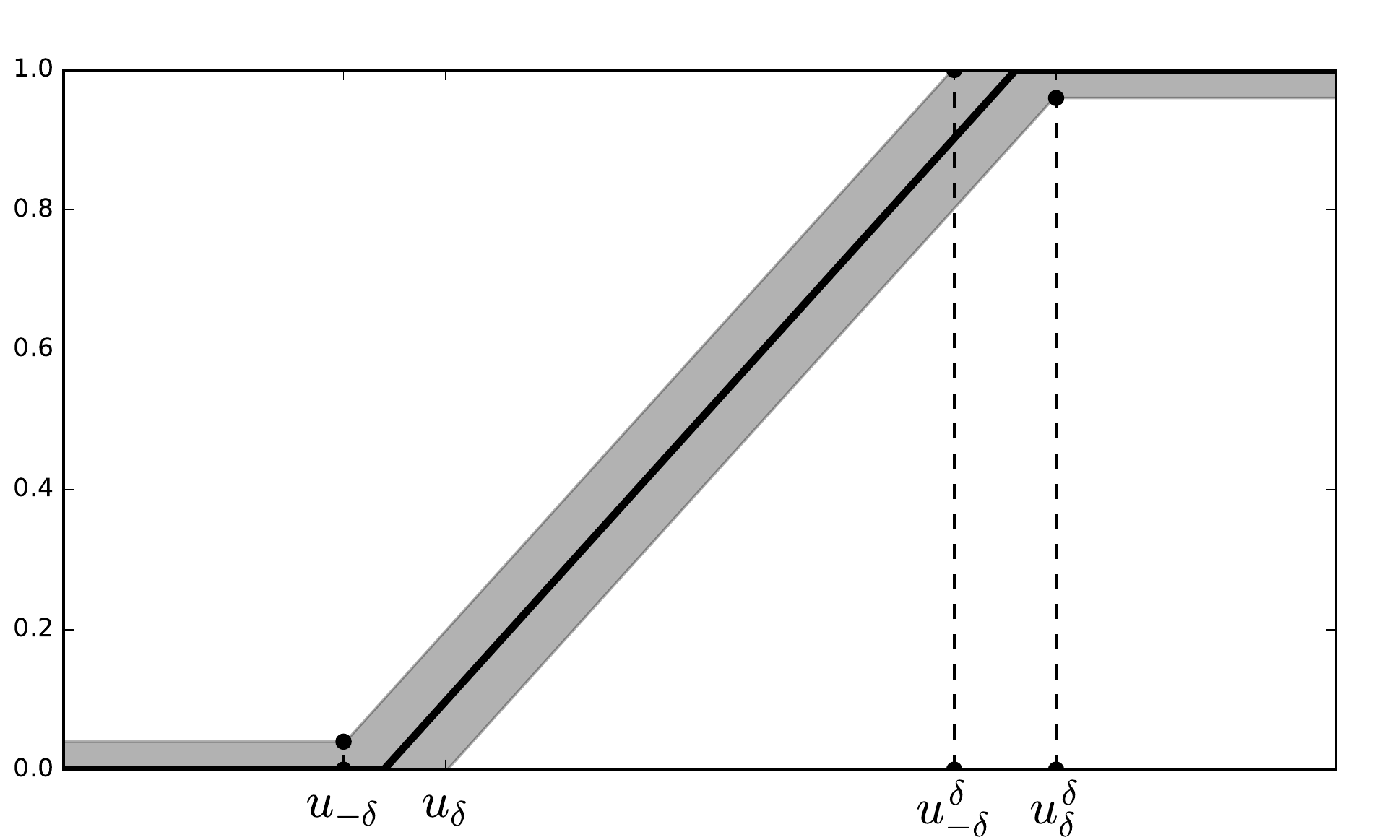}
				\caption[Illustration of $u_{-\delta}, u_{\delta}, u_{-\delta}^{\delta}$ and $u_{\delta}^{\delta}$]{Illustration of $u_{-\delta}, u_{\delta}, u_{-\delta}^{\delta}$ and $u_{\delta}^{\delta}$. The solid line represents $\mu_0$ and shaded region represents $B_L(\mu_0,\delta).$}\label{fig:udelta}
			\end{center}
		\end{figure}
		We will first prove for any $\delta_1\in (0, \delta_0)$, 
\begin{equation}
\lim_{\delta\to0}\sup_{x\in\mathcal{R}}D(\mu_x^{\delta}||B_L(P_0,\delta_0))\le D(\mu_0||B_L(P_0,\delta_0-\delta_1)). \label{eq:uppersemi}
\end{equation}
Now fix $\delta_1$, we then establish that $D(\mu_x^{\delta}||B_L(P_0,\delta_0))$ can be uniformly bounded as $x$ varies. Denote $$P_{\delta_0-\delta_1}:=\arg\inf_{\{P\in B_L(P_0,\delta_0-\delta_1)\}}D(\mu_0||P).$$
		For fixed $\delta<\delta_1$, let $$P_{\delta_0-\delta_1}^{\delta,u}(t)=(1-\delta_1)P_{\delta_0-\delta_1}(t+\delta)+\delta_1,$$
		and $$P_{\delta_0-\delta_1}^{\delta,l}(t)=(1-\delta_1)P_{\delta_0-\delta_1}(t-\delta).$$
		To get $P_{\delta_0-\delta_1}^{\delta,u}(t)$, we first shift $P_{\delta_0-\delta_1}(t)$ to the left by $\delta$, then scale it by $(1-\delta_1)$ and shift it up by $\delta_1$; similarly to get $P_{\delta_0-\delta_1}^{\delta,l}(t)$, we shift $P_{\delta_0-\delta_1}(t)$ to the right by $\delta$, then scale it by $(1-\delta_1)$. Clearly
		\begin{eqnarray}
		\label{eq:left_right_dl}
		d_L(P_{\delta_0-\delta_1}^{\delta,u},P_{\delta_0-\delta_1})\leq\delta_1,\quad 
		d_L(P_{\delta_0-\delta_1}^{\delta,l},P_{\delta_0-\delta_1})\leq\delta_1. \nonumber
		\end{eqnarray} 
		
		For any $x$, construct $P_{\delta_0-\delta_1}^x$ in a similar way as $\mu_x^{\delta}$, 
		\begin{eqnarray}
		P_{\delta_0-\delta_1}^x(t)=
		\begin{cases}
		P_{\delta_0-\delta_1}^{\delta,l}(t)  & \quad \text{if } t<x,\\
		P_{\delta_0-\delta_1}^{\delta,u}(t)  & \quad \text{if } t\ge x.
		\end{cases}\nonumber
		\end{eqnarray}
		$P_{\delta_0-\delta_1}^x\in B_L(P_0,\delta_0)$ since 
		\begin{eqnarray}
		d_L(P_{\delta_0-\delta_1}^x,P_0)&\le& d_L(P_{\delta_0-\delta_1}^x,P_{\delta_0-\delta_1})+d_L(P_{\delta_0-\delta_1},P_0) \nonumber\\
		&\le&\delta_1+(\delta_0-\delta_1)=\delta_0, \nonumber
		\end{eqnarray}
		where the first inequality holds because $(\mathcal{P},d_L)$ is a metric space (i.e., $d_L$ satisfies the triangle inequality); the second inequality comes from (\ref{eq:left_right_dl}) and the definition of $P_{\delta_0-\delta_1}$.
		
		From Lemma \ref{pro:finite} $D(\mu_0||P_{\delta_0-\delta_1})=D(\mu_0||B_L(P_0,\delta_0-\delta_1))<\infty$, therefore $\mu_0$ is absolutely continuous with respect to $P_{\delta_0-\delta_1}$. From the construction of $\mu_x^{\delta}$ and $P_{\delta_0-\delta_1}^x$, we can see that $\mu_x^{\delta}$ is absolutely continuous with respect to $P_{\delta_0-\delta_1}^x$ as well. Therefore, we have
		\begin{eqnarray}
\lim_{\delta\to0}\sup_{x\in\mathcal{R}}D(\mu_x^{\delta}||B_L(P_0,\delta_0)) &=&\lim_{\delta\to0}\sup_{x\in\mathcal{R}}\inf_{\{P\in B_L(P_0,\delta_0)\}}D(\mu_x^{\delta}||P) \nonumber\\ &\le&\lim_{\delta\to0}\sup_{x\in\mathcal{R}}D(\mu_x^{\delta}||P_{\delta_0-\delta_1}^x),\nonumber
		\end{eqnarray}
establishing (\ref{eq:uppersemi}).
		We now prove $D(\mu_x^{\delta}||P_{\delta_0-\delta_1}^x)$ can be uniformly bounded as $x$ varies. Inequalities appear in  cases 1)-3) are due to the log sum inequality unless otherwise stated.
		\begin{enumerate}
			\item
			For $x< u_{-\delta}$, 
			\small
\begin{eqnarray}
\ D(\mu_x^{\delta}||P_{\delta_0-\delta_1}^x) &	\leq{}&		\delta\log\frac{\delta}{\delta_1}+ \int_{u_{-\delta}}^{u_{-\delta}^{\delta}}d(\mu_0(t+\delta)+\delta)  \log\frac{d(\mu_0(t+\delta)+\delta)}{d((1-\delta_1)P_{\delta_0-\delta_1}(t+\delta)+\delta_1)} \nonumber \\
&		=&  \delta\log\frac{\delta}{\delta_1}+\int_{u_{-\delta}}^{u_{-\delta}^{\delta}}d(\mu_0(t+\delta))\log\frac{d(\mu_0(t+\delta))}{(1-\delta_1)d(P_{\delta_0-\delta_1}(t+\delta))}  \label{eq:degenerate1} \\
			&=&  \delta\log\frac{\delta}{\delta_1}+\int_{u_{-\delta}}^{u_{-\delta}^{\delta}}d(\mu_0(t+\delta))\log\frac{1}{(1-\delta_1)}
+\int_{u_{-\delta}}^{u_{-\delta}^{\delta}}d(\mu_0(t+\delta))\log\frac{d(\mu_0(t+\delta))}{d(P_{\delta_0-\delta_1}(t+\delta))} \nonumber\\
&			=& \delta\log\frac{\delta}{\delta_1}+(1-\delta)\log\frac{1}{(1-\delta_1)} +\int_{u_{-\delta}+\delta}^{u_{-\delta}^{\delta}+\delta}d(\mu_0(t))\log\frac{d(\mu_0(t))}{d(P_{\delta_0-\delta_1}(t))},  \nonumber
\end{eqnarray}
			\normalsize
			when $\delta\to0$, the above converges to $$\log\frac{1}{(1-\delta_1)}+D(\mu_0||P_{\delta_0-\delta_1}).$$
			\item
			For $u_{-\delta}\le x\le u_{\delta}$, 
			\small
\begin{eqnarray}
	D(\mu_x^{\delta}||P_{\delta_0-\delta_1}^x)
&=& (u_0(x+\delta)+\delta)\log\frac{(u_0(x+\delta)+\delta)}{(1-\delta_1)P_{\delta_0-\delta_1}(x+\delta)+\delta_1}  \nonumber\\
&& +\int_{x+}^{u_{-\delta}^{\delta}}d(\mu_0(t+\delta)+\delta)\log\frac{d(\mu_0(t+\delta)+\delta)}{d((1-\delta_1)P_{\delta_0-\delta_1}(t+\delta)+\delta_1)} \nonumber\\
&	\le& \delta\log\frac{\delta}{\delta_1}+(u_0(x+\delta))\log\frac{(u_0(x+\delta))}{(1-\delta_1)P_{\delta_0-\delta_1}(x+\delta)} \nonumber \\
&&+\int_{x+}^{u_{-\delta}^{\delta}}d(\mu_0(t+\delta))\log\frac{d(\mu_0(t+\delta))}{(1-\delta_1)d(P_{\delta_0-\delta_1}(t+\delta))} \nonumber \\
		&	\le& \delta\log\frac{\delta}{\delta_1}+\int_{u_{-\delta}}^{x}d(\mu_0(t+\delta))\log\frac{d(\mu_0(t+\delta))}{(1-\delta_1)d(P_{\delta_0-\delta_1}(t+\delta))}\nonumber \\
&&			+\int_{x+}^{u_{-\delta}^{\delta}}d(\mu_0(t+\delta))\log\frac{d(\mu_0(t+\delta))}{(1-\delta_1)d(P_{\delta_0-\delta_1}(t+\delta))} \nonumber \\
&			=& \delta\log\frac{\delta}{\delta_1}+\int_{u_{-\delta}}^{u_{-\delta}^{\delta}}d(\mu_0(t+\delta))\log\frac{d(\mu_0(t+\delta))}{(1-\delta_1)d(P_{\delta_0-\delta_1}(t+\delta))}, \label{eq:degenerate2}
			\end{eqnarray}
			\normalsize
			which degenerates to the case of $x< u_{-\delta}$ since (\ref{eq:degenerate2}) is the same as (\ref{eq:degenerate1}).
			
			\item
			For $u_{\delta}<x\le u_{-\delta}^{\delta}$,
			\allowdisplaybreaks
			\small
			\begin{eqnarray*}
		 D(\mu_x^{\delta}||P_{\delta_0-\delta_1}^x)&	=& \int_{u_{\delta}}^{x-}d(\mu_0(t-\delta)-\delta)\log\frac{d(\mu_0(t-\delta)-\delta)}{d((1-\delta_1)P_{\delta_0-\delta_1}(t-\delta))}\nonumber\\ 
&			& +(\mu_0(x+\delta)+\delta-(\mu_0(x-\delta)-\delta)) \log\frac{(\mu_0(x+\delta)+\delta-(\mu_0(x-\delta)-\delta))}{(1-\delta_1)P_{\delta_0-\delta_1}(t+\delta)+\delta_1-(1-\delta_1)P_{\delta_0-\delta_1}(t-\delta)}\nonumber\\
&&\ +\int_{x+}^{u_{-\delta}^{\delta}}d(\mu_0(t+\delta)+\delta)\log\frac{d(\mu_0(t+\delta)+\delta)}{d((1-\delta_1)P_{\delta_0-\delta_1}(t+\delta)+\delta_1)} \nonumber \\
&=& \int_{u_{\delta}}^{x-}d(\mu_0(t-\delta))\log\frac{d(\mu_0(t-\delta))}{(1-\delta_1)dP_{\delta_0-\delta_1}(t-\delta)}\nonumber\\
	&		& +(2\delta+\mu_0(x+\delta)-\mu_0(x-\delta)) \log\frac{2\delta+\mu_0(x+\delta)-\mu_0(x-\delta)}{\delta_1+(1-\delta_1)(P_{\delta_0-\delta_1}(t+\delta)-P_{\delta_0-\delta_1}(t-\delta))}\nonumber\\
&			& +\int_{x+}^{u_{-\delta}^{\delta}}d(\mu_0(t+\delta))\log\frac{d(\mu_0(t+\delta))}{(1-\delta_1)dP_{\delta_0-\delta_1}(t+\delta)} \nonumber\\
			&\le{}& \int_{u_{\delta}}^{x-}d(\mu_0(t-\delta))\log\frac{d(\mu_0(t-\delta))}{(1-\delta_1)dP_{\delta_0-\delta_1}(t-\delta)}\nonumber\\
		&	&+2\delta\log\frac{2\delta}{\delta_1}+\int_{x-\delta}^{x+\delta}d\mu_0(t)\log\frac{d\mu_0(t)}{(1-\delta_1)dP_{\delta_0-\delta_1}(t)}\nonumber\\
		&	&+\int_{x+}^{u_{-\delta}^{\delta}}d(\mu_0(t+\delta))\log\frac{d(\mu_0(t+\delta))}{(1-\delta_1)dP_{\delta_0-\delta_1}(t+\delta)} \nonumber \\
			&=&  2\delta\log\frac{2\delta}{\delta_1}+\int_{u_{\delta}-\delta}^{u_{-\delta}^{\delta}+\delta}d\mu_0(t)\log\frac{d\mu_0(t)}{(1-\delta_1)dP_{\delta_0-\delta_1}(t)}\nonumber \\
			&=& 
			2\delta\log\frac{2\delta}{\delta_1}+(1-2\delta)\log\frac{1}{1-\delta_1}+\int_{u_{\delta}-\delta}^{u_{-\delta}^{\delta}+\delta}d\mu_0(t)\log\frac{d\mu_0(t)}{dP_{\delta_0-\delta_1}(t)}, \nonumber
			\end{eqnarray*} 
			\normalsize
			when $\delta\to0$, the above converges to
			$$\log\frac{1}{1-\delta_1}+D(\mu_0||P_{\delta_0-\delta_1}).$$
			\item
			Other symmetric cases can be solved similarly.
		\end{enumerate}
		From the above arguments, we have
		\begin{eqnarray}
\lim_{\delta\to0}\sup_{x\in\mathcal{R}}D(\mu_x^{\delta}||B_L(P_0,\delta_0)) 
		&\le& \log\frac{1}{1-\delta_1}+D(\mu_0||B_L(P_0,\delta_0-\delta_1)). \nonumber
		\end{eqnarray}
		Notice that this is true for any $\delta_1$. Letting $\delta_1\to0$, we have 
		\begin{eqnarray}
\lim_{\delta\to0}\sup_{x\in\mathcal{R}}D(\mu_x^{\delta}||B_L(P_0,\delta_0)) 	&\le&\lim_{\delta_1\to0}\left( \log\frac{1}{1-\delta_1}+D(\mu_0||B_L(P_0,\delta_0-\delta_1))\right) \nonumber\\
		&=& \lim_{\delta_1\to0}D(\mu_0||B_L(P_0,\delta_0-\delta_1)) \nonumber\\
		&=&D(\mu_0||B_L(P_0,\delta_0)), \nonumber
		\end{eqnarray}
		the last equality comes from Lemma \ref{lem:kld_left_conti}: $D(\mu_0||B_L(P_0,\delta_0))$ is left continuous in $\delta_0$ if $P_0(t)$ is continuous in $t$. 
	\end{proof}

	\begin{lem}
		\label{lem:lowersemi}
		Given $P_0\in\mathcal{P}$ and $\delta_0>0$, $D(\mu||B_L(P_0,\delta_0))$ is lower semicontinuous in $\mu$ with respect to the weak convergence.
	\end{lem}
	\begin{proof}
		Assume $\mu_n\wto\mu_0$. From (\ref{eq:achieves_infimum}), we know there exists $P_n\in B_L(P_0,\delta_0)$ such that $D(\mu_n||P_n)=D(\mu_n||B_L(P_0,\delta_0))$. Since $\bar{B}_L(P_0,\delta_0)$ is compact, there exists a subsequence of $P_n$ (which we again denote by $P_n$) that converge to $P_{\mu_0}\in\bar{B}_L(P_0,\delta_0)$. $D(\mu||P_{\mu_0})\le\liminf_{n\to\infty}D(\mu_n||P_n)$ because $(\mu_n, P_n)\to(\mu_0, P_{\mu_0})$ and the KLD is lower semi-continuous. Therefore we have
		\begin{eqnarray}
		D(\mu_0||B_L(P_0,\delta_0))&=&D(\mu_0||\bar{B}_L(P_0,\delta_0)) \label{eq:050201}\\
		&\le& D(\mu_0||P_{\mu_0}) \nonumber\\
		&\le&\liminf_{n\to\infty}D(\mu_n||P_n)\nonumber\\
		&=&\liminf_{n\to\infty}D(\mu_n||B_L(P_0,\delta_0))\nonumber
		\end{eqnarray}
		where (\ref{eq:050201}) comes from (\ref{eq:exchange1}). 
	\end{proof}
	
	It is straightforward to prove $D(\mu||B_L(P_0,\delta_0))$ is lower semicontinuous in $\mu$; proving that it is also upper semicontinuous is tricky. The key step of the long proof in the this section is Lemma \ref{lem:suptor}, which is explained below.
	
	For a fixed $P_0$, with small perturbation on $\mu$, $D(\mu||P_0)$ may vary in an arbitrary manner, thus $D(\mu||P_0)$ is not upper semicontinuous. $B_L(P_0,\delta_0)$ provides the maximum freedom for tolerating the perturbation on $\mu$, since the L\'evy metric is the weakest among other metrics. For all perturbations on $\mu$ that are within $B_L(\mu,\delta)$, the largest variation of $D(\mu||B_L(P_0,\delta_0))$ is achieved by a distribution whose CDF is constructed by shifting the $\mu(t)$ both horizontally and vertically to the edge of of $B_L(\mu,\delta)$. Such shifts can be tolerated by $B_L(P_0,\delta_0)$, so as the level of perturbation on $\mu$ decreases to $0$, and the corresponding variation in $D(\mu||B_L(P_0,\delta_0))$ diminishes.
	
	By proving $D(\mu||B_L(P_0,\delta_0))$ is both upper semicontinuous and lower semicontinuous in Lemma \ref{lem:uppersemi} and \ref{lem:lowersemi}, we know that if $P_0(t)$ is continuous in $t$, $D(\mu||B_L(P_0,\delta_0))$ is continuous in $\mu$ with respect to the weak convergence.

\end{document}